\documentclass[a4paper,twoside,12pt,english]{article}

\usepackage[utf8]{inputenc}

\usepackage[english]{babel}
\makeatletter

\newcommand\dateymd{\number\year, \ifcase\month\or
January\or February\or March\or April\or May\or June\or
July\or August\or September\or October\or November\or
December\fi, \number\day}
\newcommand\printtime{%
\c@hours=\time \divide\c@hours by60
\c@minutes=\c@hours \multiply\c@minutes by-60
\advance \c@minutes by \time
\ifnum\c@hours<10 0\fi\the\c@hours:%
\ifnum\c@minutes<10 0\fi\the\c@minutes}
\makeatother

\usepackage[OT1]{fontenc}

\usepackage{fancyhdr}
\usepackage{setspace}

\fancypagestyle{footonly}{
\fancyhf{}

}

\usepackage{titlesec} 
\titleformat{\part}[block]{\large\bfseries\center}{Part \thepart}{1em}{}{\bigskip} 
\titlespacing{\part}{0pt}{1.2\bigskipamount}{1.\bigskipamount} 

\makeatletter
\renewcommand{\paragraph}{%
  \@startsection{paragraph}{4}%
  {\z@}{2.75ex \@plus .5ex \@minus 1ex}{-1em}%
  {\normalfont\normalsize\bfseries}%
}
\makeatother

\usepackage{bm} 
\usepackage{amssymb}
\usepackage{eucal}
\usepackage{amsmath}
\usepackage{amsthm}
\usepackage{enumerate}
\usepackage{enumitem}
\usepackage{caption} 
\usepackage{float} 

\usepackage[colorlinks,linkcolor=blue,citecolor=blue,urlcolor=blue,breaklinks]{hyperref}
\usepackage{breakurl}
\usepackage{natbib}

\theoremstyle{plain}
\newtheorem{lemmaN}{Lemma}

\newtheorem*{lemma}{Lemma}
\newtheorem*{theorem}{Theorem}
\newtheorem*{corollary}{Corollary}

\newtheorem*{theorem*}{Theorem}

\newcommand\conpar{\hskip-.5em}

\newcommand\remark{\smallskip\noindent\textit{Remark}.\hskip.5em}

\newcommand\ie{i.\,e.~}

\newcommand\shortquote[1]{``#1''}
\newcommand\system[1]{\textsl{#1}}
\newcommand\sigles[1]{\unskip}

\newcommand\df{\bfseries}
\newcommand\dfc[1]{\,{\df#1}\,}
\newcommand\dfd[1]{\,{\df#1}\hskip1pt}

\newcommand\ensep{\unskip\hskip.65em\ignorespaces}

\newcommand\xsmallskip{\vskip.25\smallskipamount}

\newcommand\ist{A}				
\newcommand\dft{0}				
\newcommand\winners{X} 
\newcommand\better{\!>\!}
\newcommand\aprbar{\,|\,}
\newcommand\rkrel{\mathrel{\smash{\succeq\kern-5pt\raise2pt\hbox{\mathsurround0pt${}^*$}\kern-2pt\null}}}
\newcommand\rkrelsub[1]{\mathrel{\smash{\succeq\kern-4pt\raise2pt\hbox{\mathsurround0pt${}^*$}\kern-4pt\raise-1pt\hbox{\mathsurround0pt${}_{#1}$}\kern-2pt\null}}}

\newcommand\profilex{\Omega}
\newcommand\llullx{\Gamma}

\newcommand\gerelsub[1]{\mathrel{\smash{\succeq\kern-3pt\raise-1pt\hbox{\mathsurround0pt${}_{#1}$}\kern-1.5pt\null}}}

\usepackage{colortbl}
\newcommand\skl[1]{{\,$#1$\,}}
\newcommand\labelcell[1]{\cellcolor[gray]{0.8}\makebox[.9em][c]{\skl{#1}}}
\newcommand\hlinestrut{\hline\rule{0pt}{2.5ex}}

\newcommand\onehalf{\leavevmode\hbox{\footnotesize$1\over 2$}}
\newcommand\strutoh{\rule[-7pt]{0pt}{16.5pt}}

\newcount\csjcom
 {\global\advance\csjcom by 1%
 \vskip\medskipamount\bgroup\leftskip2\parindent\noindent\hbox{\textsf{NOTA \the\csjcom.} }}%
 {\par\egroup\vskip\medskipamount}

\begin{document}

\thispagestyle{empty}
\null\vskip-20mm\null
\begin{center}
\hrule
\vskip7.5mm
\large
\textbf{\uppercase{Choosing by means of}}\par
\vskip2pt
\textbf{\uppercase{approval-preferential voting.}}\par
\vskip2pt
\textbf{\uppercase{The path-revised approval choice}}\par
\vskip4pt
\textsc{Rosa Camps,\, Xavier Mora \textup{and} Laia Saumell}
\par
Departament de Matem\`{a}tiques,\break
Universitat Aut\`onoma de Barcelona,\break
Catalonia
\par\medskip
\texttt{xmora\,@\,mat.uab.cat}
\par\medskip
31st October 2014;\,
revised 29th February 2020
\vskip7.5mm
\hrule
\end{center}

\null\vskip-7.5mm\null

\onehalfspacing

\begin{abstract}
Approval-preferential voting is problematical since it combines two different kinds of information that could by themselves lead to different choices. This article analyses the problem and studies a new proposal to deal with it.
The existing methods are overviewed with special attention to several desirable properties that are not always met. Particular emphasis is made on certain rather unknown views of Condorcet about this subject. The proposed procedure definitely aims for a choice in the spirit of approval voting. However, it still makes use of the preferential information so as to improve the approval one. This is done by means of the path scores, in common to Schulze's method for preferential voting. The resulting method, that we call path-revised approval choice, is shown to enjoy several good properties. In particular, it fulfils in a well-defined sense Condorcet's latest view that a surely good option should prevail over a doubtfully best one.
\end{abstract}
\vskip3mm


\rightline{Il meglio è nemico del bene}
\rightline{(Orlando Pescetti, 1603, \textit{Proverbi Italiani}\rlap{)}}

\vskip10mm\hrule\vskip10mm
\newpage

\onehalfspacing


\null\vskip-2cm\null 
\setcounter{section}{-1}
\section{Introduction}

Suppose that two alternative proposals $a$ and $b$
are submitted to a decision-making body formed by 100~members.
Each member is asked whether he approves of $a$ or not, likewise about $b$, and whether he prefers $a$ to $b$ or viceversa.
Assume that the votes are as follows:
\begin{equation}
\label{eq:ex1}
25:~a \aprbar b\,,\quad 
35:~b \better a \aprbar,\quad
40:~|\, b \better a\,,
\end{equation}
where the numbers mean quantities of voters,\,
$x \better y$ means that $x$ is preferred to $y$,\,
and a bar indicates that the options at its left are approved and those at its right are disapproved.
In other words, \eqref{eq:ex1} means that 25~voters approve of $a$ and disapprove of $b$,
35~prefer $b$ to $a$ while approving of both, and 40 disapprove of both while still preferring $b$ over $a$.
Which decision should be taken as a result of these votes? 

Notice that proposal~$a$ is approved by a majority of 60\%, whereas $b$ is disapproved by a majority of 65\%.
From this point of view, it makes sense to go for proposal~$a$.
However, one can also check that a majority of 75\% prefers $b$ to $a$, which advises against choosing~$a$.
Should one let the preferential information prevail? Or should it be dismissed in favour of the approval information?

The situation is very much that of a Condorcet cycle of collective preferences~(see for instance \citealt[\S\,2.3]{laslier:2004}).
In~fact, approving a proposal can be interpreted as preferring it to a certain implicit option that we shall refer to as default.
This default option can be identified with the \textit{status quo}, \ie leaving things as they are.
Alternatively, one could think of it as an abstract boundary between good and bad options.
Anyway, under such an interpretation
the approval-preferential vote \eqref{eq:ex1} about $a$ and $b$
can be viewed as the following purely preferential vote about $a$, $b$ and $\dft$,
where $\dft$ denotes the default option:
\begin{equation}
\label{eq:ex1dft}
25:~a \better \dft \better b\,,\quad 
35:~b \better a \better \dft\,,\quad
40:~\dft \better b \better a\,.
\end{equation}
From this perspective, 
the remarks made in the preceding paragraph correspond indeed to a Condorcet cycle,
namely $a \better \dft \better b \better a.$ 
Such cycles are not just an academic possibility, but they have occurred in practice,
as in the two-choice referendum that was held in the canton of Bern on 28~November 2004~(\citealt{bochsler:2010}; see also \S\,\ref{par:bern2010} below).

This point of view leads to consider the different methods that have been devised for dealing with preferential voting
and the problem of cyclic collective preferences (see for instance \citealt[Ch.\,13]{tideman}).
However, most if not all of these methods abide by the premise of neutrality, \ie an equal treatment of all options,
which does not apply to the present case. In fact, be it the status quo or an abstract boundary between good and bad options,
the default option plays a special role that justifies a different treatment.

So the question remains of how to make a choice through approval-preferen\-tial voting.
In this article we will deal with this problem in a general setting.
We will start by overviewing the existing methods. This is done in Section~1,
which is also aimed at pointing out several desirable properties that are not always satisfied.
This will motivate a method that made its appearance in \citealt[\S\,7.2]{chr} (in a different but equivalent formulation)
and whose study is deepened here.
We will refer to it and its outcomes ---possibly more than one in the event of certain ties---
as the path-revised approval choice(s) (\citealt{chr} use the terms `goodness method' and `goodness winners').
This method will be developed in Section~2. As we will see, it has several good properties,
most of which are obtained for the first time in this article.
One of these properties is very much in the spirit of Condorcet’s last views on elections
---rather unknown and in conflict with his celebrated earlier principle---
namely that \emph{a surely good option should prevail over a doubtfully best one}.
The mathematical proofs of the new properties are given in appendix~A.

By the way, as we will see in \S{2.3}, in the particular case of \eqref{eq:ex1} the path-revised approval choice will be the default option, i.\,e.~the \textit{status quo}.

\section{Existing methods and desirable properties}

\subsection{The Swiss multiple-choice referendum procedures}
\label{ssec:Swiss-procedure}

An interesting real example of approval-preferential voting is provided by multiple-choice referenda as they are conducted in the Swiss Confederation and its cantons. More specifically, approval-preferential voting arises in two cases: (a)~popular initiatives, in which case the government can put forward a counter-proposal; and~(b)~the so-called `constructive' referenda ---in use in only a few cantons---
where a proposal from the government can be followed by one or more counter-proposals from groups of voters.

In these cases, the voters are asked two sets of questions. The main set is about every proposal by itself to see whether the voter approves it or not. For a proposal to be adopted it must receive more approvals than disapprovals. In the event of more than one proposal satisfying this condition, a choice is made on the basis of the answers to the second set of questions, where the voter is asked to express his preferences between the different proposals.

In the more frequent case of two proposals, the second part reduces to a single question, namely which of the two proposals is preferred to the other. This information determines which option is chosen when both proposals satisfy the condition of having more approvals than disapprovals. This procedure was put forward by Christoph Haab (\citeyear{haab:1976}) and was adopted at the federal level in 1987 (see \citealt[p.\,165]{lutz-strohmann:1998}).

In example~\eqref{eq:ex1} this procedure chooses proposal~$a$, since it is the only one with more approvals than disapprovals. In fact, $a$ has 60~approvals and 40~disapprovals, whereas $b$ has 35~approvals and 65~disapprovals.

For more than two proposals ---which may be the case in constructive referenda---
it can happen that three or more of the proposals satisfy the condition of having more approvals than disappovals.
In this case, the problem arises of how to choose between them on the basis of the preferential information.
In the cantons of Bern and Nidwalden, this is done by means of the Copeland rule (\citealt[p.\,206--209]{tideman}) restricted to the set of approved options.
More specifically, the choice is made on the basis of the number of victories
against the other approved options ($x$ is considered to beat $y$ when there are more voters who prefer $x$ to $y$ than viceversa).
If~necessary, a tiebreaker is also used.
In Bern, the tiebreaker is the total number of favourable votes in the preference comparisons against all the other proposals \cite[Art.\,139.7]{bern:2012}.
In Nidwalden, the tiebreaker is the number of approvals minus the number of disapprovals \cite[Art.\,44.3]{nidwalden:1997}.
Such procedures will be referred to as \textsl{Swiss Procedures}.

\subsection{The ranking approach: rank all options, including the default one, then take the top one}
\label{ssec:rank-all-options}

Another interesting real case of approval-preferential voting is the voting procedure used by the Debian Project~(\citeyear{debian}; see also \citealt{voss:2012}).

\smallskip
Since 1998, the votes of this organization systematically include a default option to mark the boundary between approved and disapproved options. In fact, their regulations explicitly state that \shortquote{Options which the voters rank above the default option are options they find acceptable. Options ranked below the default options are options they find unacceptable} (\citealt[v\,1.1, \S A.6]{debian}). In practice, the default option is usually described as ``further discussion'' or ``none of the above''.

By the way, the idea of systematically including such an option was already clearly stated by Charles L. Dodgson, alias Lewis Carroll, in the following way:
``where `no Election' is allowable, the phrase should be treated exactly as if it were the name of a candidate'' (\citealt[Chap.\,II]{dodgson:1873}).

\smallskip
Going back to the Debian Project, the procedure that they use to make a choice ---possibly the default option--- is due to Markus Schulze (\citeyear{schulze:2003,schulze:2011}) who introduced it in 1998 in a mailing list about election methods. As it is shown by that author, this procedure admits of several equivalent formulations. One of them corresponds to actually ranking all the options and then taking the top one. We will call this option the \textsl{Path-Top Choice}, since this ranking is based on the so-called \textit{path scores}. These scores will be dealt with in detail in Section~\ref{ssec:path-scores}, for they will also be a crucial tool of our proposal.

Instead of that particular ranking method, one can consider any other method for selecting a best or topmost option in a preferential vote, with the only peculiarity that one of the options is the default one. Among such methods, one can consider the \textsl{Borda} count \sigles{B}, the method of \system{Condorcet-Kem\'eny-Young} \sigles{CKY}, or the method of \system{Ranked Pairs} \sigles{RP}. As a general reference for these and other methods, we refer the reader to \citealt[Ch.\,13]{tideman}. 

Except for the Borda count, all of the other methods mentioned in the two preceding paragraphs comply with the Condorcet principle. That is, they choose the Condorcet winner whenever it exists. Recall that a Condorcet winner means an option~$x$ that beats every other option~$y$ in the sense that a majority of voters prefer $x$ to $y$.  

For three options (one of them being the default one) and fully completed ballots,
all of the above-mentioned Condorcet-compliant methods amount to resolving any Condorcet cycle by dropping the weakest, \ie less supported, of the three majoritarian views in conflict.
In~the case of example~(\ref{eq:ex1}), this means dropping the view of approving $a$,
which leads to adopting the default option~$\dft$.

So these methods allow strongly supported preferences to overturn the approval information.

\subsection{Should a small preference differential prevail over a large approval differential?}
\label{ssec:should-a-small}

However, we might be giving too much importance to preferences. 
Consider, for instance, the following example (from \citealt[eq.\,(109)]{chr}):
\begin{equation}
\label{eq:ex4e}
\onehalf+\epsilon: \ a \better b \aprbar,\qquad \onehalf-\epsilon: \ b \aprbar a\,,
\end{equation}
where the numbers of voters are normalized to add up to one
and $\epsilon$ is a small positive quantity
(before the normalization the votes could be, for instance,
$500\,001: \ a \better b \aprbar$ and $500\,000: \ b \aprbar a$).
As one can see, both $a$ and $b$ are approved ---\ie preferred to $\dft$--- by a majority of voters;
besides, $a$~is preferred to $b$ also by a majority.
So $a$ is a Condorcet winner, and therefore it will be chosen by the above Condorcet-compliant methods.
However, the majority that approves $a$ and prefers it to $b$ is quite slight,
whereas $b$ is approved by a whole unanimity.
So, we are allowing a tiny preference differential to overcome a huge approval differential.

As one can easily check, the Swiss procedures also choose $a$.

Examples like this suggest that
one should perhaps completely forget about preferences
and take into account only the approval information.

For later reference, we will refer to this course of action as the \system{Approval Choice} procedure.
More precisely, we will take the general view that not approving an option is not the same as disapproving it,
in which case it is appropriate to abide by the number of approvals minus the number of disapprovals
and to choose the option that maximizes this differential (see for instance \citealt{richelson:1984},
under `differential maximization', as well as \citealt{felsenthal:1989}).
This includes the more classical particular case where the ballots allow to express approvals but not disapprovals,
in which case the above-mentioned difference reduces to the number of approvals.

Notice that in the general case where disapprovals are also expressed
that difference could be negative for all non-default options,
in which case it makes sense to declare a void choice, \ie to choose the default option.

Having said that, it still seems that there should be a sensible way
to refine the approval choice by making use of
any preferential information that could exist as a complement of the approval-disapproval one.
However, it is not clear how should one proceed to blend both kinds of information together.


\subsection{Condorcet's last views on elections}
\label{ssec:Condorcet-last-views}

The point that we are leading to 
was formulated by Condorcet in the following way \,(\citealt[\S\,XIII, p.\,307]{condorcet:1789}; \citealt[p.\,177--178]{mclean-hewitt:1994}; emphasis is ours):
\begin{quote}
\leavevmode\llap{``}%
It is generally more important to \emph{be sure} of electing men who are \emph{worthy} of holding office than to have \emph{a small probability} of electing \emph{the worthiest man}.\rlap{''}
\end{quote}
The latest works of Condorcet on voting and elections, from 1788 to his death in 1794,
are indeed dominated by this idea and by 
the aim of being able to deal with a large number of candidates,
in which case paired comparisons become rather cumbersome
(\citealt{crepel:1990,mclean-hewitt:1994}). 

Concerning the meaning of `being sure' and `probability', in another place Condorcet says 
the following \,(\citealt[\S\,XIII, p.\,193]{condorcet:1788}; \citealt[p.\,139]{mclean-hewitt:1994}):
\begin{quote}
\leavevmode\llap{``}%
We consider a proposition asserted by 15~people, say, more probable than its contradictory asserted by only 10.\rlap{''}
\end{quote}
So he identifies the probability of a proposition ---such as ‘$x$ is worthy’--- with the support for it in the vote.

%

In the same spirit as example~(\ref{eq:ex4e}), Condorcet gives the following one
(\citealt[p.\,34--35]{crepel:1990}; \citealt[p.\,241]{mclean-hewitt:1994}):
\begin{equation}
\label{eq:exCondorcet}
5: a \better c \better ... \aprbar ...,\qquad 4: b \better c \better ... \aprbar ...,
\end{equation}
where he assumes a large number of candidates and, although he does not use approval bars, he explicitly says that all voters consider $c$ worthy of the place.
So~$a$ is considered the worthiest candidate by a slight majority, but $c$ is considered worthy by unanimity,
and so it might be a more appropriate choice than $a$
(especially if the latter were disapproved by the second set of voters).
By the way, this and other examples
show that Condorcet was accepting the possibility of making a choice different from the Condorcet winner.

\subsection{Condorcet's practical methods}
\label{ssec:Condorcet-practical-methods}

The methods that Condorcet proposed in connection with the preceding ideas
are often referred to as Condorcet's ``practical'' methods.
In general terms, there are two of them.
In both cases, 
the voter is required to produce an ordered list of approved candidates.
Unlike proper approval voting, however,
the length of this list is fixed (\citealt[p.\,203]{condorcet:1788}; \citealt[p.\,143]{mclean-hewitt:1994}):
\begin{quote}
\leavevmode\llap{``}%
It should not be too short, to give a good chance that one of the candidates will obtain a majority, [...] nor should it be too long, [so that] the voters can still complete the list without having to nominate candidates they consider unworthy.\rlap{''}
\end{quote}

In his first practical proposal, formulated in 1788, 
Condorcet chooses the most approved candidate,
conditioned to having obtained a majority,
and the preferential information is used only in the event of ties
(\citealt[Article~V, p.\,193--211]{condorcet:1788}; translated in \citealt[p.\,139--147]{mclean-hewitt:1994}).
If no candidate has a majority of approvals, then he simply proposes to run a second round after having asked the voters to extend their lists with a certain number of additional candidates.
So this proposal was very much in the spirit of approval voting.


In his second and final practical proposal,
formulated in 1789 and yet, with some variations, in 1793,
Condorcet makes a more substantial use of the preferential information
(\citealt{condorcet:1789,condorcet:1793}).
More specifically, his last work 
contains the following wording (within a more complex multiround procedure;
\citealt[p.\,41--42, \S\,{VI}]{crepel:1990}; \citealt[p.\,249--250, \S\,{VI}]{mclean-hewitt:1994}):
\begin{quote}
\leavevmode\llap{``}%
If one candidate has the absolute majority of first votes, he will be elected.\ensep
If one candidate has the absolute majority of first votes and second votes together, he will be elected. If several candidates obtain this majority, the one with the most votes will be preferred.\ensep
If one candidate has the absolute majority of the three votes together, he will be elected, and if several candidates obtain this majority, the one with the most votes will be preferred.\rlap{''}
\end{quote}
By the context it is clear that the number of first (resp.\ second) votes means the number of ballots where that candidate is ranked first (resp.\ second). Here Condorcet had limited the preferential vote to three candidates, but in the following round [\textit{ibidem}, \S\,{VIII}] he extends this rule to preferential votes that list six candidates.


Except for secondary variations, this idea spread and/or was rediscovered several times.
Shortly after Condorcet's proposal, it was adopted in Geneva, where it was analyzed by Simon Lhuilier (\citeyear{lhuilier:1794}).
Later on, in~the beginning of the twentieth century it was adopted by several American municipalities,
starting from Grand Junction (Colorado) in 1909 under the leadership of James W.~Bucklin (see \citealt{bucklin}).


A few years later, the essentials of the Bucklin system were brought out by Clarence G.~Hoag
(\citeyear{hoag:1914}; see also \citealt[\S\,278]{hoag-hallett:1926}).
In particular, he assumed each voter to give a truncated strict ranking of an arbitrary length
(instead of Bucklin's limitation to a maximum of three grades
with the possibility of several options tied to each other in the third grade).
In such a setting, the rule to determine the winner remains as it is stated in the above quote of Condorcet,
except for its going on beyond third options.

Much more recently, this procedure has been rediscovered again by Remzi Sanver (\citeyear{sanver:2010})
who calls it `Majoritarian Approval Compromise'.

On~account of its origins, we will refer to it as the \system{Condorcet-Bucklin} \sigles{CB} procedure.
Unfortunately, the recent literature has not been giving due credit to Condorcet
(unlike \citealt[p.\,7]{hoag:1914}, and \citealt[p.\,490]{hoag-hallett:1926}).
Besides, it is often stated that the Bucklin system requires complete rankings,
which is not the case of the original works cited above.

Truncated rankings can certainly be seen as a form of approval-preferential voting.
In this connection, it will usually be appropriate to understand that
all the options that are included in a truncated ranking meet the approval of the voter.
Depending of the case, it may be appropriate or not to interpret also that these are the only approved options
(see \S\,\ref{ssec:ties-and-incomplete}).
When the voters are obliged to produce complete rankings
one can still apply the Condorcet-Bucklin procedure,
but the notion of approval is left out.




Under this interpretation,
the normal outcome of the Con\-dor\-cet-Bucklin procedure will have the property of being approved by a majority of the voters.
In fact, it will be so unless no option has the property of appearing in a majority of the votes.
In this case, we will take the view that the outcome is not a proper option, but the default one.

When applying the Condorcet-Bucklin procedure to such cases as \eqref{eq:ex1} and \eqref{eq:ex4e},
we will first convert the approval-preferential votes into truncated rankings
by throwing away any information that might appear at the right of the bars.

As one can easily check,
in the case of example~(\ref{eq:ex1}), the Condorcet-Bucklin procedure chooses option~$a$,
which is the only one that is approved by a majority.
In example~(\ref{eq:ex4e}) with $\epsilon>0$ it chooses also option~$a$,
in spite of the fact that this option is approved only by a slight majority,
whereas $b$ is approved by unanimity.
So the Condorcet-Bucklin procedure does not fulfil Condorcet's own last thinking (\S\,\ref{ssec:Condorcet-last-views})
that clearly good options should be favoured over doubtfully good ones.

\vskip-8mm\null 

\subsection{The Instant Runoff procedure and the issue of monotonicity}
\label{ssec:apv-instant-runoff}

Another real example where truncated strict rankings are used 
is the Irish presidential election. In 1937 the Constitution of Ireland established that this election shall use the one-seat version of the single transferable vote. 
This procedure is known as \system{Instant Runoff} \sigles {IR} or \system{Alternative Vote} (\citealt[\S\,277]{hoag-hallett:1926}; \citealt[p.\,193--195]{tideman}). 
It~proceeds by successive elimination. At every step, each of the options that stay on is evaluated by counting how many ballot papers show that option in first position. If this count is larger than one half of the non-empty ballot papers, the option in question becomes chosen. Otherwise, the option with the lowest count is eliminated and the same process is repeated after having cancelled this option on all ballot papers (which may cause some of them to become empty). 

At present, this procedure is also used in the elections to the Legislative Assemblies of two states\,/\,territories from Australia, namely, New South Wales, since 1980, and the Northern Territory, since 2016, both of which use one-seat constituencies.


As it is well known, a major flaw of this system is its lack of \emph{monotonicity}
(see for instance \citealt[p.\,194]{tideman}). 
More specifically, they can exhibit the following behaviour:
An option $a$ that is chosen from a certain set of ballots can cease to be chosen if one ballot
of the form $b > a\,...$ is changed into $a > b\,...$, where the dots stand for exactly the same content in both cases.
So $a$ ceases to win because of having been raised in one of the ballots! This is certainly quite undesirable.


One could argue that monotonicity is not important in practice, because a real election is about a particular set of ballots and not about any variation of them. However, a lack of monotonicity can be seen as an evidence that the results are somehow wrong, either before the variation or afterwards (or in both cases).


\subsection{Upper semicontinuity: changes imply ties}
\label{ssec:usc}

Like monotonicity, the issue 
that we are about to raise is also about varying the profile of a vote.
By the \dfc{profile} of a vote we mean a specification of how many ballots were cast of each possible kind,
that is, a specification of the form of \eqref{eq:ex1} or \eqref{eq:ex4e}
(on the understanding that no ballots have been cast of any other kind than shown).
Instead of giving the absolute frequencies, as in~\eqref{eq:ex1}, one can give the relative ones, as in~\eqref{eq:ex4e}.

Let us look at the effect of varying $\epsilon$ in example \eqref{eq:ex4e}. As we have seen, the Swiss procedures choose option $a$ for $\epsilon > 0$ and it is easily checked that they choose $b$ for $\epsilon \le 0$. We claim that this is not right: if the result must change from $a$ to $b$, which are clear choices respectively for $\epsilon=\onehalf$ and $\epsilon=-\onehalf,$
then somewhere in between both $a$ and $b$ should have the same merit, and therefore, a tie between them should be admitted for that value of~$\epsilon$.

This condition is easily violated by the methods that successively apply different criteria, like the Swiss procedures.
In an attempt to fix it, one can consider modifying them by replacing proper majority requirements by weak ones (greater than or equal to 50\%).
This achieves the desired result in the particular case of \eqref{eq:ex4e}, where one gets $\{\,a,b\,\}$ for $\epsilon = 0.$
However, the problem persists in other examples, like the following one: 
\begin{equation}
\label{eq:ex4}
\begin{alignedat}{3}
2+\epsilon&:\,a \better b \better c \aprbar,\quad
&3&:\,c \better a \better b \aprbar\,,\quad
&4&:\aprbar b \better c \better a\,,\\[-2pt]
1-\epsilon&:\,a \better b \aprbar c\,,\quad
&1&:\,a \better c \better b \aprbar\,,\qquad
&1&:\,b \aprbar c \better a\,.
\end{alignedat}
\end{equation}
One can check that each of the Swiss procedures of \S\,\ref{ssec:Swiss-procedure} chooses $a$ for small negative $\epsilon$ and $b$ for small positive $\epsilon$; at the boundary it chooses $a$ or $b$ depending on whether proper majority or weak majority is considered, but neither of both variants ever chooses $\{a,b\}$. So neither of them satisfies the condition that we are talking about.


In general terms, this condition can be formulated in the following way, where the output of a choice procedure is allowed to consist of several options tied to each other. Consider a particular profile and the set of choices $X$ 
that corresponds to it according to the method in question; then, for any other profile close enough to that one ---\ie for close enough frequencies--- the corresponding set of choices should be a subset of $X$. In mathematical terms, such a property is known as `upper semicontinuity'.

One can argue that we are considering the profile as a continuous variable instead of a discrete one, as it is the case in practice. However, and similarly to the above view on monotonicity, the significance of upper semicontinuity is not a matter of practical incidence, but one of consistency with respect to variations of the profile.

\subsection{Taking into account the preferences between non-approved options}
\label{ssec:pba}

The methods mentioned in Sections~\ref{ssec:Condorcet-practical-methods} and \ref{ssec:apv-instant-runoff} 
do not take into account the preferences that a voter could have between his non-approved options. 
This is unfair towards the voters who do not approve at all the chosen option and would rather prefer some other of their non-approved options.

In order to take into account all preferences one could certainly use the methods of Section~\ref{ssec:rank-all-options}
after having augmented the set of options with the default one.
By the way, one can include among them the Condorcet-Bucklin method for complete rankings
(which in the case of~(\ref{eq:ex1}) chooses neither $a$ nor $0$, but $b$!).
However, as we raised in Sections~\ref{ssec:should-a-small} and \ref{ssec:Condorcet-last-views}, this approach is too preference-oriented;
instead, one should give some sort of priority to the approval information.

The existing proposals in this direction are essentially some more elaborated versions of the Swiss procedures that we presented in Section~\ref{ssec:Swiss-procedure}.

One of them is the \system{Preference-Approval Voting} \sigles{PAV} procedure that was put forward by Steven Brams and Remzi Sanver (\citeyear{brams-sanver:2009}; see also \citealt{brams}).
When two or more options are approved by a majority, this procedure restricts the attention to these options
and the preferential information about them is used to single out, if possible, their Condorcet winner;
if this is not possible, then
all options in their so-called top cycle are considered and their approval winner is chosen.
(For a given subset $S$ of options ---here we are considering the majority-approved ones---
its top cycle is the minimal $X\subseteq S$ such that every $x\in X$ beats any $y\in S\setminus X$
in the sense that a majority of voters prefer $x$ to $y$; see for instance \citealt[\S\,{3.3} and 3.5]{brill:2016}).


A simpler possibility is the \system{Approval Voting with a Runoff} \sigles{AVR}, considered by Remzi Sanver~(\citeyear{sanver:2010}). Here, the preferences are used only to compare between the two most approved options.


Anyway, in the case of~(\ref{eq:ex4e}) both these procedures keep choosing~$a$ for arbitrarily small $\epsilon> 0$, like the Swiss procedures, against the view expressed in Sections~\ref{ssec:should-a-small} and~\ref{ssec:Condorcet-last-views} that an overwhelming approval for an option should prevail over a slightly majoritarian preference for another.

\subsection{Dealing with ties and with incomplete information}
\label{ssec:ties-and-incomplete}

In practice, voters often do not have an opinion on some options. Besides, they may also rank equally some of them that they do know. In order to properly deal with such possibilities, one must begin by  distinguishing between them when interpreting the ballots.

In this connection, the Debian voting rules state that \shortquote{Unranked options are considered to be ranked equally with one another}~(\citealt[\S A.6.1]{debian}). However, this is really questionable.
In Condorcet's words (\citealt[p.\,194]{condorcet:1788}; \citealt[p.\,139]{mclean-hewitt:1994}):
\begin{quote}
\leavevmode\llap{``}%
When someone votes for one particular candidate, he simply asserts that he considers that candidate better than the others, and makes no assertion whatsoever about the respective merits of these other candidates. His judgement is therefore incomplete.\rlap{''}
\end{quote}

A vote where two options $x$ and $y$ are really ranked equally with each other can be
assimilated to half a vote where $x$ is preferred to $y$ together with half another vote with the reverse preference.
In contrast, a vote that expresses neither a preference nor a tie between $x$ and $y$ should contribute neither to the number of voters who prefer $x$ to $y$ nor to the number of those who prefer $y$ to $x$.

Usually one takes the view that \shortquote{Ranked options are considered preferred to all unranked options}~(\citealt[\S A.6.1]{debian}). However, in some contexts it could be more appropriate to interpret that no comparison is made between a ranked option and an unranked one.

One should also be aware that not approving an option is not the same as really disapproving it.

On account of all these considerations, it is certainly desirable that the ballots be designed so as to make as clear as possible what the voter really means to say (no matter whether he is being sincere or not). Besides, it is most important to clearly specify how will the ballots be interpreted.

Once the information has been properly interpreted and collected, the problem remains of how should one deal with it. In fact, many existing methods assume that complete information is given,
and quite often it is not clear at all how should they be extended to the general incomplete case.


\subsection{The path heuristic}

Let us look once more at example \eqref{eq:ex1}. Should $a$ be approved or not?
We certainly have an argument to approve $a$: Simply, this view is supported by 60\% of the voters, certainly more than the 40\% who support the contrary view.  However, we can still produce an argument for disapproving $a$: In~fact, the statement that $a$ is less preferable than $b$ is supported by 75\% of the voters, and the statement that $b$ is disapproved, i.\,e.~less preferable than the default option, is supported by 65\% of the voters. So we have two statements with the following properties: their conjunction entails disapproving $a$, and both of them are supported more strongly than that 60\% in favour of approving~$a$.
Somehow this provides an argument for disapproving $a$ with a ‘strength’ of $\min(75\%,65\%) = 65\% > 60\%$. Of course, we should allow for any analogous argument that might help towards approving $a$. However, in this particular example no such argument
improves upon the original support of 60\%. So we end up with $a$ being disapproved by 65\% to 60\%.


This idea can be elaborated into a general method for the so-called problem of ``judgment aggregation'', that is, for collectively deciding about several issues under the requirement of satisfying a given set of logical constraints (\citealt{dp}).
Having said that, here we will pay special attention to a particular subset of issues, namely
which options get approved or not, and to the strength of the corresponding decisions.
As we shall see, this leads to a procedure, the \system{Path-Revised Approval Choice} \sigles{PRAC},
that meets all the requirements that we have been pointing out in the preceding sections.

\section{The path-revised approval choice}
\label{sec:prac}

\subsection{The original preference scores}

We assume a finite set of options~$\ist$.
The outcome of the procedure will be a subset of~$\ist$.
Ordered pairs of options will be denoted by juxtaposition;
so $xy$ means the ordered pair whose first element is $x$ and whose second element is $y$.


Our starting point are the \dfd{preference scores}, a.\,k.\,a.~paired-comparison scores,
that is the numbers of voters who expressed a preference for $x$ over~$y$,
where $x$ and $y$ vary over all ordered pairs of \emph{different} options.
These numbers will be denoted by $V_{xy}$.

Consider, for instance, the case of example~(\ref{eq:ex1}). As we have been saying, a proposal $x$ being approved amounts to it being preferred to a certain default option $\dft.$ By taking this into account, the profile~(\ref{eq:ex1}) results in the preference scores that are shown below in~\eqref{eq:ex1-final-scores}. Since we are interested only in
pairs $xy$ with $x\neq y$, we use the diagonal cells for specifying the simultaneous labelling of rows and columns by the existing options; so the cell located in row $x$ and column $y$ gives information about the preference of $x$ over $y$.
\begin{equation}
\label{eq:ex1-final-scores}
(V_{xy}) = 
\begin{small}
\begin{tabular}{|c|c|c|}
\hlinestrut
\labelcell{a}&25&60\\
\hlinestrut
75&\labelcell{b}&35\\
\hlinestrut
40&65&\labelcell{\dft}\\
\hline
\end{tabular}
\end{small}
\,.
\hskip.75em
\end{equation}

In the event of ties or incomplete information, the translation of the ballots into the preference scores should take into account the remarks that were made in Section~\ref{ssec:ties-and-incomplete}.

The final results will not depend on the absolute magnitude of the preference scores, but only on their relative magnitude.
Therefore, it doesn't matter if the numbers of voters are normalized to add up to $1$, as it was the case in \eqref{eq:ex4e}.


The whole collection of preference scores will be called the \dfc{Llull matrix} of the vote (after Ramon Llull, who already made use of it in the 13th century, see \citealt[Ch.\,1, Section 4.2 and Ch.\,3]{mclean-urken:1995}).

\subsection{The path-revised preference scores}
\label{ssec:path-scores}

Our proposal is crucially based upon the \dfc{path-revised preference scores}, more shortly known as \dfc{path scores}. They can be seen as an upwards revision of the original preference scores in the light of the logical implications that are associated with the notion of transitivity.

More specifically, for every pair of options $x$ and $y$, one considers all paths from $x$ to $y$, i.\,e.~all finite sequences $\gamma = (x_0, x_1, ... x_m)$ with $x_0 = x$ and $x_m = y$. Each such path is postulated to provide an indirect support in favour of $x$ being preferred to~$y$, namely the minimum of the numbers $V_{x_ix_{i+1}}\ (0\le i<m)$. The maximum of all these indirect supports is then taken as the revised support in favour of $x$ being preferable to $y$.

This revised support, or path score, is therefore defined by
\begin{equation}
\label{eq:vstar}
V^*_{xy} = \max\, \{\,V_\gamma \mid \gamma \text{ is a path } (x_0, x_1, ... x_m) \text{ from } x_0=x \text{ to } x_m=y\,\},
\end{equation}
where
\begin{equation}
\label{eq:vgamma}
V_\gamma = \min_{0\le i<m} V_{x_ix_{i+1}}.
\end{equation}
One easily checks that in \eqref{eq:vstar} one can restrict $m$ to be less than or equal to the number of options.
In fact, a longer path will necessarily contain some repeated option, and deleting the portion between repetitions
will always result in a higher or equal score.

\smallskip 
For three options the path scores are easily obtained by hand.
For instance, in the case of the Llull matrix \eqref{eq:ex1-final-scores} 
one easily finds that $V^*_{ab}=\max\,(V_{ab},\min(V_{a0},V_{0b})) = \max\,(25,\min(60,65)) = 60.$
By proceeding in the same way for all pairs of options, the table of path scores for that example is found to be as follows:
\begin{equation}
\label{eq:ex1-final-revised-scores}
(V^*_{xy}) = 
\begin{small}
\begin{tabular}{|c|c|c|}
\hlinestrut
\labelcell{a}&60&60\\
\hlinestrut
75&\labelcell{b}&60\\
\hlinestrut
65&65&\labelcell{\dft}\\
\hline
\end{tabular}
\end{small}
\,.
\hskip.75em
\end{equation}

For more than three options, the computation of the path scores becomes a combinatorial matter. However, there is a standard procedure for it, namely the Floyd-Warshall algorithm (\citealt[\S\,{25.2}]{cormen:2009}) whose computing time grows only as $N^3$, where $N$ stands for the number of options.



\smallskip 
Once the path scores have been computed, now it is a matter of comparing their values for opposite pairs of options, such as $xy$ and $yx$. \emph{If $V^*_{xy}$ is larger than $V^*_{yx}$ then we will adopt as a collective opinion the view that $x$ is preferable to $y$.} Furthermore, we will measure the \dfc{confidence} of this collective opinion by the difference $V^*_{xy}-V^*_{yx}$ (relative to the total number of voters).

For instance, table~\eqref{eq:ex1-final-revised-scores} leads to adopt the view that $b$ preferable to $a$ with a confidence of~$75-60=15,$
and that $\dft$ is preferable to both $a$ and $b$ with a confidence of~$65-60=5.$

\medskip
Recall that the very same rule applied to the original preference scores~\eqref{eq:ex1-final-scores} was leading to a cycle, namely $a$ collectively preferred to $\dft$, $\dft$ collectively preferred to $b,$ and $b$ collectively preferred to $a$.
In contrast, the collective preferences that we infer from \eqref{eq:ex1-final-revised-scores} are transitive.
Now, this is not a casual fact, but a general property of the path scores:

\begin{theorem*}[Schulze, 1998]
\label{st:transSchulze}
If $V^*_{xy} > V^*_{yx}$ and $V^*_{yz} > V^*_{zy}$ then $V^*_{xz} > V^*_{zx}$.
\end{theorem*}

\noindent
This fundamental fact was pointed out in 1998 by Markus Schulze,
who gave a proof of it in a mailing list about election methods. 
The proof can be found also in \citep{schulze:2003,schulze:2011,crc}.

\medskip
In the following we shall write $x \succ y$ for $V^*_{xy}>V^*_{yx}$.
In the absence of any ties of the form $V^*_{xy}=V^*_{yx}$, 
the transitive relation~$\succ$
completely ranks all the options.
A natural choice is certainly the option that heads this complete ranking
(like option $\dft$ in our running example). 
In the following we will refer to it as the \system{Path-Top Choice}.

In the presence of ties, the matter becomes more involved,
which leads to a variant that we will refer to as \system{Strong Path-Top Choice}.

More precisely and generally, the path-top choices are those options $x$ such that $x\rkrel y$ for any $y\neq x,$
whereas the strong path-top choices are those options $x$ such that $x\succeq y$ for any $y\neq x.$
Here $x\succeq y$ stands for $V^*_{xy}\ge V^*_{yx}$
and $\rkrel$ means the transitive closure of the relation $\succeq$ (which need not be transitive).
A more detailed account is given in \S\,\ref{ssec:osi-ptc}. 

\medskip
However, our aim is not a choice of this sort,
but one in the spirit of approval voting.

\subsection{A choice in the spirit of approval voting}
\label{ssec:the-prac}

From now on we single out one of the options to play a special role.
This special option will be called \dfc{default} and it will be denoted by $\dft$.
By definition, a non-default option $x$ being approved means that it is preferred to~$\dft$.
Similarly, $x$ being disapproved means that $\dft$ is preferred to~$x$.


Instead of looking for the best option in comparison with all the others,
now we aim for the best one in comparison with the default,
\ie the best one in terms of approval and disapproval.

Let $x$ be a non-default option. In principle, the approval and disapproval information about $x$ is given by the scores $V_{x\dft}$ and $V_{\dft x}$. 
However, in the case of approval-preferential voting
one can improve upon this information by turning to 
the path scores $V^*_{x\dft}$ and $V^*_{\dft x}.$
In fact, these numbers are the result of revising the approval and disapproval scores in the light of additional information that is provided by the preference scores between non-default options.

On the lines of the preceding section, an option $x$ will be granted approval as a collective opinion whenever the difference $V^*_{x\dft}-V^*_{\dft x}$ is positive.
Besides, the larger this margin, the higher the confidence of this opinion. Therefore, it makes sense to choose an option $x$ that maximizes this confidence, \ie that satisfies
\begin{equation}
\label{eq:def-rac}
V^*_{x\dft}-V^*_{\dft x} \,\ge\, V^*_{y\dft}-V^*_{\dft y},\qquad
\text{for any $y\in\ist\setminus \{x\}.$}
\end{equation}
In the following we will refer to such an option as a \dfc{path-revised approval choice,}
and the set formed by all of them will be denoted by $\winners.$

In the preceding paragraph we have assumed the inequality
$V^*_{x\dft}-V^*_{\dft x} > 0$ to hold for at least one non-default option $x$.
Such a situation will be referred to as the \dfc{proper} case.
In the improper case, that is when $V^*_{x\dft}-V^*_{\dft x} \le 0$ for all $x$,
we will take as path-revised approval choice
the default option together with any non-default option $x$ that might exist with $V^*_{x\dft}-V^*_{\dft x} = 0$.
In~this case, the options contained in $\winners$
are not properly approved, but they
belong to the boundary between approval and disapproval.

\vskip2\smallskipamount 
For our running example~\eqref{eq:ex1}, the revised Llull matrix \eqref{eq:ex1-final-revised-scores}
results in $V^*_{a\dft}-V^*_{\dft a} = V^*_{b\dft}-V^*_{\dft b} = -5.$ So, the path-revised approval choice is the default option ---the status quo---
in contrast to the Swiss procedures, which choose proposal~$a$.

In this case the path-revised approval choice is the same as the path-top choice.
In other cases they are not the same.
For instance, in \S\,\ref{ssec:ex4e} we shall see that
for example~\eqref{eq:ex4e} with $0<\epsilon<\onehalf$ the path-revised choice is the unanimously approved option~$b$,
whereas the path-top choice is~$a$.

\medskip
Under the name of `goodness method', 
\citealt[\S\,7]{chr}, followed a slightly different approach where the notion of approval-disapproval was dealt with 
by itself, without any reference to a default option. Besides, the approval and disapproval scores were revised only in the light of such implications as ``$x$ being good and $y$ being preferable to $x$ implies $y$ being good'', and similarly, ``$x$ being bad and $x$ being preferable to $y$ implies $y$ being bad'', leaving aside any other implications associated with the transitivity of preferences.
Nevertheless, it turns out that the revised approval and disapproval scores in that setting are exactly the same as in the present one. This equivalence allows to transfer here the monotonicity result obtained in that article.

In the following section this result is complemented with several other properties of the path-revised approval choice.

\vskip-8mm\null 

\subsection{Properties of the path-revised approval choice}
\label{ssec:properties}

This section briefly goes over the properties of the path-revised approval choice.
Except for the first item, that requires no proof,
and the property of monotonicity, whose proof is found elsewhere, 
the other results are new and their mathematical proofs are given here
---in Appendix~A--- for the first time.

\paragraph{General domain.}
To begin with, the path-revised approval\linebreak choice
can deal with any sort of incompleteness of information or ties in the ballots. 
In particular, it takes into account any preferences of the voter between 
his non-approved options.
See, for instance, the example in \S\,\ref{par:lastexample}.

\paragraph{Monotonicity
\textnormal{(proof in \citealt[Thm.\,7.1 and Cor.\,7.2]{chr};\, see \S\,\ref{ssec:osi-monotonicity})}.}
\label{ssec:monotonicity}
If $x$ is a path-revised approval choice 
and some votes are modified by solely raising $x$ to a better position,
then $x$ remains a path-revised approval choice.

More generally, the same conclusion holds for any modification of the votes whereby
the preference scores of the form $V_{xz}$ increase or stay the same,
those of the form $V_{zx}$ decrease or stay the same,
and any other preference score stays the same.

\paragraph{Upper semicontinuity \textnormal{(proof in \S\,\ref{ssec:osi-usc})}.} 
The set of path-revised approval choices is an upper semicontinuous function of the profile
in the sense of Section~\ref{ssec:usc}.
In particular, in the case of two non-default options the path-revised approval choice cannot change from one option to the other 
without somewhere in between admitting both of them.

\paragraph{Condorcet's last desideratum \textnormal{(proof in \S\,\ref{ssec:osi-cld})}.} 
\label{ssec:Condorcet-last-desideratum}
As we saw in Section~\ref{ssec:Condorcet-last-views}, Condorcet's last views on elections
insisted on making sure that a good option is chosen \,rather than\, aiming for the best option but not being so sure about it.
The path-revised approval choice fulfils this goal in the following way,
where the term `confidence' should be understood in the sense of \S\,\ref{ssec:path-scores}:
if $x$ is a path-revised approval choice and $y$ is not,
then the confidence that $x$ is preferable to the default option
is greater than or equal to the confidence that $y$ is preferable to $x$;
besides, this inequality is strict as soon as the path-revised approval choice is a proper one
(\ie it satisfies the strict inequality $V^*_{x0} - V^*_{0x} > 0$).


%

\paragraph{Weak Pareto consistency \textnormal{(proof in \S\,\ref{ssec:osi-pareto})}.} 
\label{ssec:pareto}
Assume that the votes have the form of truncated rankings (not necessarily strict) and that they are interpreted as it is mentioned in Section~\ref{ssec:ties-and-incomplete}.
Assume also that two different options $x$ and $y$ satisfy the following condition: every voter either prefers $x$ to $y$ or ranks them equally, but at least one voter strictly prefers $x$ to $y$. In that case, $y$ cannot be a path-revised approval choice unless $x$ is also a path-revised approval choice.

Instead of this property, one 
might have expected the stronger one that forbids any such $y$ to be a path-revised approval choice.
For future reference we will refer to this stronger version as \dfc{strong Pareto consistency.}

The reason why the  path-revised approval choice satisfies only the weaker version is because it aims only at making sure that we are choosing 
an option preferable to the default one; the existing preferences between non-default options are considered only to the extent that they can contribute to that effect.
This is clearly illustrated by example~\eqref{eq:ex4e} with $\epsilon=\onehalf$, where all voters agree that both~$a$ and~$b$ are approved and that $a$ is preferable to~$b$ (i.\,e.\ all votes have the form  $a > b > 0$). In this case, the preferential information does not entail any revision of the view that both~$a$ and~$b$ are unanimously approved. So, both~$a$ and~$b$ remain path-revised approval choices in spite of the fact that $a$ definitely dominates~$b$ in the sense of Pareto.


\paragraph{Comparison with other methods.}
\label{ssec:comparison}
Table~\ref{table:comparison} below compares the Path-Revised Approval Choice with other methods that we have been mentioning in Section~1.
Comparison is made in terms of a few selected properties.
Proofs and counterexamples are given or referenced in Appendix~A. 

\newcommand\si{\raise2pt\hbox{\small$\surd$}}
\newcommand\no{$\times$}
\makeatletter
\renewcommand\strut{\vrule\@height14pt\@depth6pt\@width\z@}
\makeatother
\begin{table}[h]
\vskip5mm
\begin{center}
\begin{small}
\begin{tabular}{ l | c | c | c | c | c |}
\strut																			& Type & Mono & USC & \!\!Pareto\!\! & CLD \\ \hline
\strut \textsl{Approval Choice} \,(\S\,\ref{ssec:should-a-small})					& ad  & \si  & \si & w  & \no \\ \hline
\strut \textsl{Condorcet-Bucklin} \,(\S\,\ref{ssec:Condorcet-practical-methods})	& tr  & \si  & \no & s  & \no \\ \hline
\strut \textsl{Instant Runoff} \,(\S\,\ref{ssec:apv-instant-runoff}) 				& tr  & \no  & \si & s  & \no \\ \hline
\strut \textsl{Path-Top Choice}	\,(\S\,\ref{ssec:path-scores} and \ref{ssec:osi-ptc})
																				& g   & \no  & \si & w  & \no \\ \hline
\strut \textsl{Strong Path-Top Choice} (\S\,\ref{ssec:path-scores} and \ref{ssec:osi-ptc})
																				& g   & \si  & \si & w  & \no \\ \hline
\strut \textsl{Swiss Procedures} \,(\S\,\ref{ssec:Swiss-procedure})		& g   & \si  & \no & s  & \no \\ \hline
\strut \textsl{Preference-Approval Voting} \,(\S\,\ref{ssec:pba})					& g   & \si  & \no & w  & \no \\ \hline
\strut \textsl{Approval Voting with a Runoff} \,(\S\,\ref{ssec:pba})				& g   & \si  & \no & w  & \no \\ \hline
\strut \textsl{Path-Revised Approval Choice} \,(\S\,\ref{sec:prac})				& g   & \si  & \si & w  & \si \\ \hline
\end{tabular}
\end{small}
\captionsetup{width=85mm}
\caption{Comparison of different methods for making a choice by approval-preferential voting.}
\label{table:comparison}
\end{center}
\null\vskip-5mm 
\end{table}

The first column of that table divides the methods into different `types' depending on the kind of information that they assume and deal with:  
`ad' means that only approval-disapproval information is used (any preferential information is thrown away);
`tr'~means that the ballot is constrained to have the form of a truncated strict ranking of non-default options;
`g'~stands for the general setting, meaning that the method can deal with any information that can be represented by means of a Llull matrix
after having explicitly included the default option
---in particular, these methods take into account the preferences of a voter between his non approved options.

The column `Mono' indicates which methods comply with monotonicity.
The statement of this property given in the first part of \S\,\ref{ssec:monotonicity} assumes that the ballot has the form of a truncated ranking.
However, for the methods of general type, monotonicity can be understood in the more general form stated in the second part of \S\,\ref{ssec:monotonicity}.

The column `USC' indicates which methods comply with upper semicontinuity as explained in \S\,\ref{ssec:usc}.
In this connection we recall that we are dealing with the whole set of possible choices that can be made,
which may be more than one because of the ties that may arise in the procedure
(as in the Instant Runoff eliminations, or when selecting the two most approved options in the Approval Voting with a Runoff).

The next column is about the properties of Pareto consistency considered in \S\,\ref{ssec:pareto}:
`w' means the weak version and `s' means the strong one.
In this connection, the fact that both Condorcet-Bucklin and Instant Runoff satisfy the strong version
is partly due to the fact that they require the votes to be \emph{strict} truncated rankings.

Finally the `CLD' column emphasizes the fact that property \S\,\ref{ssec:Condorcet-last-desideratum} of the Path-Revised Approval is in the spirit of Condorcet's last desideratum (\S\,\ref{ssec:Condorcet-last-views}) and has no counterpart in the other methods.

\subsection{Additional examples}

\paragraph{}
\label{ssec:ex4e}
Let us consider the vote \eqref{eq:ex4e}, where $\epsilon$ is allowed to vary in the interval $-\onehalf \le \epsilon \le \onehalf$.
The Llull matrix and the associated margins are as follows:
\begin{equation*}
\label{eq:ex4e-final-scores}
(V_{xy})\!= 
\begin{small}
\begin{tabular}{|c|c|c|}
\hlinestrut
\labelcell{a}\strutoh&$\onehalf+\epsilon$&$\onehalf+\epsilon$\\
\hlinestrut
$\onehalf-\epsilon$&\labelcell{b}\strutoh&$1$\\
\hlinestrut
$\onehalf-\epsilon$&$0$&\labelcell{\dft}\strutoh\\
\hline
\end{tabular}
\end{small}
\,,\quad
(V_{xy}\!-\!V_{yx})\!= 
\begin{small}
\begin{tabular}{|c|c|c|}
\hlinestrut
\labelcell{a}\strutoh&$2\epsilon$&$2\epsilon$\\
\hlinestrut
$-2\epsilon$&\labelcell{b}\strutoh&$1$\\
\hlinestrut
$-2\epsilon$&$-1$&\labelcell{\dft}\strutoh\\
\hline
\end{tabular}
\end{small}
\,.
\hbox to0pt{\quad(11)\hss}
\end{equation*}
\setcounter{equation}{11}%
One easily checks that the only path score that differs from the original score is $V^*_{\dft b}$, which is equal to $\min(\onehalf-\epsilon,\onehalf+\epsilon),$ that is $\onehalf-\epsilon$ for $\epsilon\ge0$ and $\onehalf+\epsilon$ for $\epsilon\le0$. So the path scores and associated margins are as follows:
\begin{equation*}
\label{eq:ex4e-final-revised-scores}
(V^*_{xy})\!= 
\begin{small}
\begin{tabular}{|c|c|c|}
\hlinestrut
\labelcell{a}\strutoh&$\onehalf+\epsilon$&$\onehalf+\epsilon$\\
\hlinestrut
$\onehalf-\epsilon$&\labelcell{b}\strutoh&$1$\\
\hlinestrut
$\onehalf-\epsilon$&$\onehalf-|\epsilon|$&\labelcell{\dft}\strutoh\\
\hline
\end{tabular}
\end{small}
\,,\quad
(V^*_{xy}\!-\!V^*_{yx})\!= 
\begin{small}
\begin{tabular}{|c|c|c|}
\hlinestrut
\labelcell{a}\strutoh&$2\epsilon$&$2\epsilon$\\
\hlinestrut
$-2\epsilon$&\labelcell{b}\strutoh&$\onehalf+|\epsilon|$\\
\hlinestrut
$-2\epsilon$&$-\onehalf\!-\!|\epsilon|$&\labelcell{\dft}\strutoh\\
\hline
\end{tabular}
\end{small}
\,.
\hbox to0pt{\quad(12)\hss}
\end{equation*}
\setcounter{equation}{12}%

Therefore, the unanimously approved option $b$ is a path-revised approval choice for any $\epsilon\in[-\onehalf,+\onehalf]$ and $a$ is also included as a path-revised approval choice for $\epsilon = \onehalf$. These results are exactly the same as those of the Approval Choice, in accordance with the margins $V_{a\dft}-V_{\dft a}$ and $V_{b\dft}-V_{\dft b}$ that are shown in~(11). 
For $\epsilon\ge0$ these choices differ from the path-top ones, which reduce to $a$ for $\epsilon>0$ and admit both $a$ and $b$ for $\epsilon = 0.$
Remarkably enough, any other method that appears in Table~1 chooses option $a$ for $\epsilon>0$.


\paragraph{}
\label{par:bern2010}
Let us look now at the Bern\,2004 referendum mentioned in the Introduction. 
The data given in \citealt[Table~1]{bochsler:2010} contain a misprint in one of the figures.
The official data, which are published in \citealt{bern:2004}, are as follows, where $a,b,0$ stand respectively for the amendment of the parliament, the people's amendment, and the status quo:
\begin{equation*}
\label{eq:bern-final-scores}
(V_{xy})\!= 
\begin{small}
\begin{tabular}{|c|c|c|}
\hlinestrut
\labelcell{a}&101\,586&109\,812\\
\hlinestrut
106\,863&\labelcell{b}&104\,144\\
\hlinestrut
102\,796&106\,832&\labelcell{\dft}\\
\hline
\end{tabular}
\end{small}
\,,\quad
(V_{xy}\!-\!V_{yx})\!= 
\begin{small}
\begin{tabular}{|c|c|c|}
\hlinestrut
\labelcell{a}&-5\,277&7\,016\\
\hlinestrut
5\,277&\labelcell{b}&-2\,688\\
\hlinestrut
-7\,016&2\,688&\labelcell{\dft}\\
\hline
\end{tabular}
\end{small}
\,.
\hbox to0pt{\quad(13)\hss}
\end{equation*}
\setcounter{equation}{13}%
The total number of voters was $225\,758,$ which is larger than any of the numbers $V_{xy}+V_{yx}$.
This is due to the fact that some voters did not answer all the questions.
Anyway, the approval choice is again~$a$. Besides, it is the only proposal with more approval than disapproval,
which entails that it is also the choice of the Swiss procedures.

In this case, the path scores and their margins are as follows:
\begin{equation*}
\label{eq:bern-final-revised-scores}
(V^*_{xy})\!= 
\begin{small}
\begin{tabular}{|c|c|c|}
\hlinestrut
\labelcell{a}&106\,832&109\,812\\
\hlinestrut
106\,063&\labelcell{b}&106\,063\\
\hlinestrut
106\,063&106\,832&\labelcell{\dft}\\
\hline
\end{tabular}
\end{small}
\,,\quad
(V^*_{xy}\!-\!V^*_{yx})\!= 
\begin{small}
\begin{tabular}{|c|c|c|}
\hlinestrut
\labelcell{a}&769&2\,980\\
\hlinestrut
-769&\labelcell{b}&-769\\
\hlinestrut
-2\,980&769&\labelcell{\dft}\\
\hline
\end{tabular}
\end{small}
\,.
\hbox to0pt{\quad(14)\hss}
\end{equation*}
\setcounter{equation}{14}%
According to the signs of these margins, $a$ is also the path-revised choice as well as the path-top one.

\paragraph{}
\label{par:lastexample}
Next we give an example where the path-revised approval choice is neither the approval choice nor the path-top one.
This example will also illustrate how to deal with general truncated rankings.
Specifically, the votes are as follows:
\begin{equation}
\label{eq:ex-final}
\begin{split}
9:\, a \better b \better c \better d \better 0\,,\qquad
1:\, b \better a \better c \better d \better 0\,,\quad\\[-1pt]
1:\, d \better 0\,,\qquad\hskip-2pt
5:\, a \better d \better 0 \better b \better c\,,\qquad 
9:\, c\,.\quad 
\end{split}
\end{equation}

Like Debian elections (see \S\,\ref{ssec:ties-and-incomplete}), we will interpret that any unranked option is less preferred than any ranked one. However, we will not infer anything about the preference between two unranked options. In particular, the vote ``$d \better 0$'' says only that $d$ is approved and any other non-default option is disapproved, whereas the last vote, limited to mentioning $c$, says that $c$ is approved but it does not say anything about the approval or disapproval of the non-default options other than $c$.

By applying these rules, we get the absolute preference scores and margins that are shown next.
\begin{equation*}
\label{eq:ex-final-scores}
(V_{xy})\!= 
\begin{small}
\begin{tabular}{|c|c|c|c|c|}
\hlinestrut
\labelcell{a}&14&15&15&15\\
\hlinestrut
1&\labelcell{b}&15&10&10\\
\hlinestrut
9&9&\labelcell{c}&19&19\\
\hlinestrut
1&6&6&\labelcell{d}&16\\
\hlinestrut
1&6&6&0&\labelcell{\dft}\\
\hline
\end{tabular}
\end{small}
\,,\quad
(V_{xy}\!-\!V_{yx})\!= 
\begin{small}
\begin{tabular}{|c|c|c|c|c|}
\hlinestrut
\labelcell{a}&13&6&14&14\\
\hlinestrut
-13&\labelcell{b}&6&4&4\\
\hlinestrut
-6&-6&\labelcell{c}&13&13\\
\hlinestrut
-14&-4&-13&\labelcell{d}&16\\
\hlinestrut
-14&-4&-13&-16&\labelcell{\dft}\\
\hline
\end{tabular}
\end{small}
\,.
\hbox to0pt{\quad(16)\hss}
\end{equation*}
\setcounter{equation}{16}%

By looking at the preference scores, we see that $V_{ax} > V/2 = 25/2$ for every $x\neq a$.
So $a$ is a Condorcet winner.
However, by inspecting the margins we see that the approval choice is $d,$
since this option realises the maximum value
of the margin over~$\dft,$ namely~16.

In order to work out the path-revised approval choice and the path-top one,
we compute the path scores
and their corresponding margins. The resulting values are as follows:
\begin{equation*}
\label{eq:ex-final-revised-scores}
(V^*_{xy})\!= 
\begin{small}
\begin{tabular}{|c|c|c|c|c|}
\hlinestrut
\labelcell{a}&14&15&15&15\\
\hlinestrut
9&\labelcell{b}&15&15&15\\
\hlinestrut
9&9&\labelcell{c}&19&19\\
\hlinestrut
6&6&6&\labelcell{d}&16\\
\hlinestrut
6&6&6&6&\labelcell{\dft}\\
\hline
\end{tabular}
\end{small}
\,,\quad
(V^*_{xy}\!-\!V^*_{yx})\!= 
\begin{small}
\begin{tabular}{|c|c|c|c|c|}
\hlinestrut
\labelcell{a}&5&6&9&9\\
\hlinestrut
-5&\labelcell{b}&6&9&9\\
\hlinestrut
-6&-6&\labelcell{c}&13&13\\
\hlinestrut
-9&-9&-13&\labelcell{d}&10\\
\hlinestrut
-9&-9&-13&-10&\labelcell{\dft}\\
\hline
\end{tabular}
\end{small}
\,.
\hbox to0pt{\quad(17)\hss}
\end{equation*}
\setcounter{equation}{17}%

Since there are no ties of the form $V^*_{xy}=V^*_{yx}$, the transitive relation $\succ$ ($x\succ y$ when $V^*_{xy} > V^*_{yx}$) is a complete ranking, namely $a\succ b\succ c\succ d\succ \dft.$
So the path-top choice is~$a$.
In~contrast, the path-revised approval choice is~$c$, which realises the maximum value for the revised margin over $\dft,$ namely~$13.$

So, from the point of view of the path scores,
option $a$ looks better than any other,
but the best one in comparison with the default is option~$c$.

Notice also that the margin $V^*_{c\dft}-V^*_{\dft c}$
is not only larger than $V^*_{x\dft}-V^*_{\dft x}$ for any other option $x$,
but it is also larger than $V^*_{xc}-V^*_{cx}$. 
In other words, our confidence that $c$ is a good option ---\ie preferable to the default option--- 
is stronger than our confidence about any other option $x$ being preferable to $c$.
As we have pointed out in \S\,\ref{ssec:Condorcet-last-desideratum},
this is a general property of the path-revised approval choice
in the spirit of Condorcet's last desideratum.

\section{Concluding remarks}

As we have seen in Section~1, some organizations do take their decisions by means of approval-preferential voting. That is, each voter can express himself both in terms of approval or disapproval of separate options and in terms of preferences between them.


These two kinds of information
complement each other quite suitably
from the voters' point of view.
However, it is not clear which procedure should be used to determine the decision from a given set of votes.
Some of the existing procedures give priority to the approval information,
while others pay more attention to the preferential information.
The two approaches may easily lead to different choices.
So the question of how to combine both kinds of information is not a trivial one. 

In this connection, it makes sense to take the view that approving an option amounts to preferring it to the \textit{status quo}, and that disaproving it amounts to the contrary preference. Therefore, one can look at approval-preferential voting as simple preferential voting about an augmented set of options. Having said that, one must be ready to drop 
the usual requirement of neutrality, since the \textit{status quo} option plays a special role, which justifies a different treatment.

More specifically, in this context there are at least two different concepts of choice:
\begin{enumerate}[nosep,leftmargin=2cm]
\item[(A)] The best option in comparison with the \textit{status quo}.
\item[(B)] The best option in comparison with all the others
(including the \textit{status quo}).
\end{enumerate}


Like the ordinary approval choice, 
the path-revised approval choice is definitely aimed only at~(A).
However, the existing information about the preferences between non-default options is not thrown away.
Instead, it is used for a preliminary revision of the whole set of approval-preferential scores.
This revision is done by means of the so-called path scores, 
whose computation follows the transitivity implications.

As we have seen, the path-revised approval choice
allows to deal with any sort of incompleteness of information or ties in the ballots,
which easily arises in practice.
On the other hand, it enjoys the properties of monotonicity, upper semicontinuity and weak Pareto consistency.

Finally, the path-revised approval choice fulfils in a well-defined sense Condorcet's last desideratum
of choosing a surely good option rather than a doubtfully best one.

\appendix
\section{Mathematical proofs and other technicalities}

\noindent
The first section of this appendix is devoted to a precise definition of the methods that we call
path-top choice and strong path-top choice.
After that, we will go over the different properties that are mentioned in Section~\ref{ssec:properties}.
For each property, we will substantiate its fulfilment by the path-revised approval choice,
and then we will check for it on the other methods that appear in Table~1 of Section~\ref{ssec:comparison}.

\subsection{The path-top choices}
\label{ssec:osi-ptc}

The path scores can be used to make a choice not only in the spirit of approval voting,
but also in the spirit of looking for the best option in comparison with all the others,
without any special attention to a possible default option. In fact, the latter was the standpoint
of the original proposal made in 1998 by Markus Schulze (\citeyear{schulze:2003,schulze:2011})
on the basis of path scores.

\smallskip
So the next developments still rely on the framework of Section~\ref{ssec:path-scores}.
More specifically, our starting point is the preference relation $\succ$ defined by
\begin{equation}
\label{eq:relsucc}
x\succ y  \quad\text{if and only if}\quad V^*_{xy} > V^*_{yx}.
\end{equation}

\smallskip
In accordance with the theorem of Section~\ref{ssec:path-scores}, 
this asymmetric relation is ensured to be transitive.
In the absence of ties of the form $V^*_{xy}=V^*_{yx}$, this relation is also complete;
so in this case $\succ$ totally orders all the options
and therefore defines a unique topmost option.

\smallskip
Generally speaking, however, the transitive relation $\succ$ need not be complete,
which makes things a little more involved.
To begin with, one is led to consider
the relation $\succeq$ defined by the non-strict inequality $V^*_{xy} \ge V^*_{yx}.$
This relation is certainly complete. However, it need not be transitive.
This leads to consider its transitive closure~$\rkrel,$
i.\,e.~$x\rkrel y$ if and only if
there exists a path $x_0x_1\dots x_m$ from $x_0=x$ to $x_m=y$ such that
$V^*_{x_ix_{i+1}} \ge V^*_{x_{i+1}x_i}$ for all $i<m.$
In contrast to $\succ$ and $\succeq$, the relation $\rkrel$ is ensured to be both transitive and complete.


\bigskip
Since $\rkrel$ is a complete ranking,
a natural choice is provided by its topmost options,
that is, the members of the set
\begin{equation}
T = \{\,x\in A \mid x \rkrel y,\ \forall y \neq x\,\}.
\end{equation}
These options are what we call the \dfd{path-top choices}.

\smallskip
In the literature on tournament solutions such a set appears under the names of weak top cycle,
GETCHA set (for GEneralized Top-CHoice Assumption), or Smith set 
(see \citealt{schwartz:1986} and \citealt[\S\,3.5]{brill:2016}).
\emph{Notice, however, that we are considering it for the tournament which arises from the path scores~$V^*$,
and not for that which arises directly from the original scores~$V$.}

\smallskip
The next statement characterizes the set of path-top choices
directly in terms of the relation $\succ$ instead of $\rkrel$.
In this connection, we use the following terminology:
A set $X$ of options is a \dfc{dominant set for~$\succ$}
if and only if it is not empty and it satisfies
$x\succ y$ (i.\,e.\ $V^*_{xy} > V^*_{yx}$) for all $x\in X$ and $y\not\in X$.
Besides, $X$ is a \dfc{minimal dominant set for~$\succ$}
if and only if it is a dominant set for $\succ$
and no proper subset of $X$ has this property.

\newcommand\bla{\citealt[Cor.\,6.2.1 and 6.2.2]{schwartz:1986}}
\begin{theorem}[\bla]
\label{st:topranked-smith}
The set of path-top choices is the only minimal dominant set for~$\succ$.
\end{theorem}

%

\begin{corollary}
\label{st:topranked-strict}
An option $x$ is the only path-top choice if and only if it satisfies $V^*_{xy} > V^*_{yx}$ for all $y\neq x.$
\end{corollary}


\bigskip
Instead of the weak top cycle of $\succ$, one can consider its strong top cycle,
a.\,k.\,a.~GOCHA set (for Generalized Optimal CHoice Axiom), or Schwartz set 
(see \citealt{schwartz:1986} and \citealt[\S\,3.5]{brill:2016}).
In the present context where $\succ$ is transitive, this set amounts to
\begin{equation}
\label{eq:gocha}
S = \{\,x\in A \mid x \succeq y,\ \forall y \neq x\,\}.
\end{equation}

\smallskip
Obviously, $S \subseteq T,$ with equality whenever $\succeq$ is transitive.
When $\succeq$ is not transitive, then it does not qualify as a complete ranking,
so the members of $S$ are doubtfully acceptable as the only topmost options.
In exchange, however, the set $S$ performs better than $T$ in some respects
(such as monotonicity, as we will see in \S\,\ref{ssec:osi-monotonicity}),
which justifies considering $S$ as an alternative choice set.
In particular, this is the choice set adopted by Markus Schulze (\citeyear{schulze:2003,schulze:2011}).
Here we will refer to it as the set of \dfd{strong path-top choices}.


\subsection{Monotonicity}
\label{ssec:osi-monotonicity}

For ballots that have the form of truncated rankings (not necessarily strict), 
a choice method being monotonic means the following:
If $x$ is a choice by that method for a given profile, then it remains a choice when some votes of that profile
are modified by raising $x$ to a better position without any other change. More specifically,
``raising $x$ to a better position'' means one or more of the following changes or any sequence of them:
(a)~changing $z > x$ to $x > z$ for some~$z$; 
(b)~changing $z > x$ to $x \sim z$ for some~$z$; 
(c)~changing $x \sim z$ to $x > z$ for some~$z$; 
(d)~changing from $x$ not being mentioned in a truncated ranking to it being appended at the end of it.

For the general setting, where the input can be any Llull matrix,
the preceding concept admits the the following natural generalization:
$x$ remains a choice for any modification of the votes whereby
the preference scores of the form $V_{xz}$ increase or stay the same,
those of the form $V_{zx}$ decrease or stay the same,
and any other preference score stays the same.

\bigskip
\renewcommand\bla{\citealt[Thm.\,7.1 and Cor.\,7.2]{chr}}
\begin{theorem}[\bla]
\label{st:prac-monotonicity}
The path-revised approval choice is monotonic.
\end{theorem}


\remark
The standpoint of \citealt[\S\,7]{chr}, is a priori slightly different from the one adopted here.
However, both approaches lead to exactly the same revised approval and disapproval scores
---our $V^*_{x\dft}$ and $V^*_{\dft x}$, there denoted by $V^*(g_x)$ and $V^*(\overline g_x)$ and given by formulas (106) and (107)---.
This allows to transfer the above monotonicity result to the present setting.

\paragraph{\textsl{Approval Choice} (\S\,\ref{ssec:should-a-small}) is monotonic.} \conpar
An approval choice is an option $x$ that maximizes the value of $V_{x\dft}-V_{\dft x}$.
Obviously, it will remain chosen if $V_{x\dft}$ increases and/or $V_{\dft x}$ decreases. 

\paragraph{\textsl{Condorcet-Bucklin} (\S\,\ref{ssec:Condorcet-practical-methods}) is monotonic.} \conpar
See \citealt[p.\,204]{tideman}.

\paragraph{\textsl{Instant Runoff} (\S\,\ref{ssec:apv-instant-runoff}) is not monotonic.} \conpar
See \citealt[p.\,194]{tideman}. 

\paragraph{\textsl{Path-Top Choice} (\S\,\ref{ssec:osi-ptc}) is not monotonic.} \conpar
The following counterexample was already given in \citealt[\S\,12]{crc}:
\begin{equation}
\label{eq:ex2008} 
\begin{split}
1:\, a \better d \better b \better e \better c,\quad
1:\, b \better a \better c \better e \better d,\quad
1:\, b \better c \better a \better d \better e,\quad\\[-2pt]
1:\, b \better c \better d \better e \better a,\quad
1:\, b \better e \better c \better a \better d,\quad
1:\, d \better a \better b \better c \better e,\quad\\[-2pt]
2:\, e \better a \better c \better d \better b,\quad
1:\, e \better c \better a \better d \better b,\quad
1:\, b \better d \better c \better a \better e,\quad
\end{split}
\end{equation}
The resulting preference scores and path scores are as follows:
\begin{equation*}
\label{eq:ex2008-final-scores}
(V_{xy})\!= 
\begin{small}
\begin{tabular}{|c|c|c|c|c|}
\hlinestrut
\labelcell{a}&5&5&7&5\\
\hlinestrut
5&\labelcell{b}&7&5&7\\
\hlinestrut
5&3&\labelcell{c}&7&5\\
\hlinestrut
3&5&3&\labelcell{d}&5\\
\hlinestrut
5&3&5&5&\labelcell{e}\\
\hline
\end{tabular}
\end{small}
\,,\quad
(V^*_{xy})\!= 
\begin{small}
\begin{tabular}{|c|c|c|c|c|}
\hlinestrut
\labelcell{a}&5&5&7&5\\
\hlinestrut
5&\labelcell{b}&7&7&7\\
\hlinestrut
5&5&\labelcell{c}&7&5\\
\hlinestrut
5&5&5&\labelcell{d}&5\\
\hlinestrut
5&5&5&5&\labelcell{e}\\
\hline
\end{tabular}
\end{small}
\,.
\hbox to0pt{\quad(22)\hss}
\end{equation*}
\setcounter{equation}{22}%
As one can see, there are quite a few pairs $xy$ such that $V^*_{xy}=V^*_{yx}=5$.
As a consequence, it turns out that the relation $\rkrel$
that we defined in \S\,\ref{ssec:osi-ptc} is a whole tie.
The set of path-top choices is therefore the whole of~$\ist.$
Assume now that the last ballot in \eqref{eq:ex2008} is modified
by replacing $b \better d$ by $d \better b.$
The monotonicity property requires $d$ to continue being a path-top choice.
However, it is not so. In fact, one gets $a$ as the only path-top choice.

\paragraph{\textsl{Strong Path-Top Choice} (\S\,\ref{ssec:osi-ptc}) is monotonic.} \conpar
See \citealt[\S\,{4.5}]{schulze:2011}.
The result is easily obtained from \eqref{eq:gocha} by checking that the situation $x \succeq y$
---\ie the inequality $V^*_{xy} > V^*_{yx}$--- is preserved
when the Llull entries of the form $V_{xa}$ increase or stay the same
and those of the form $V_{ax}$ decrease or stay the same.

\paragraph{The \textsl{Swiss Procedures} (\S\,\ref{ssec:Swiss-procedure}) are monotonic.} \conpar
Assume that $x$ has been chosen by one of the Swiss procedures (Bern or Nidwalden).
To start with, this requires it to be collectively aproved, \ie to satisfy $V_{x\dft} > V_{\dft x}$.
This condition is certainly preserved under the considered modifications of the Llull matrix.
The condition of maximizing the Copeland count is also preserved (\citealt[p.\,206]{tideman}).
Finally, Bern's condition of maximizing the sum $\sum_{y\neq x}V_{xy}$ will also be preserved,
as well as Nidwalden's condition of maximizing $V_{x\dft}-V_{\dft x}$.

\paragraph{\textsl{Preference-Approval Voting} (\S\,\ref{ssec:pba}) is monotonic.} \conpar
See \citealt[\S\,{5}]{brams-sanver:2009} or \citealt[\S\,{3.5}]{brams}. 
The case with ties of the form $V_{xy} = V_{yx}$ can be dealt with by means of the results of \citealt{brill:2018}.

\paragraph{\textsl{Approval Voting with a Runoff} (\S\,\ref{ssec:pba}) is monotonic.} \conpar
See \citealt[Thm.\,20.4.1]{sanver:2010}. 

\subsection{Upper semicontinuity}
\label{ssec:osi-usc}

As it was explained in Section~\ref{ssec:usc}, this property is about small variations of the profile.
Recall that the profile of a vote means a specification of the frequency 
of every possible way to fill in the ballot.
So we can represent it as an element of~$\mathbb R_+^K$, where $K$ is the number of possible ways to fill in the ballot.

Let $X(U)\subseteq\ist$ denote the choice set associated with the profile $U$ by a given method.
This method being upper semicontinuous means that every profile $U$ has a neighbourhood $\cal N$ in $\mathbb R_+^K$ such that
$\tilde U\in \cal N$ implies $X(\tilde U) \subseteq X(U)$.

\smallskip
In the following we will often make use of an equivalent formulation.
For every option $x\in\ist,$ let $\profilex(x)$ denote the set of profiles $U$ such that $X(U)$ contains $x$.
When $\ist$ is finite, as in our case, the upper semicontinuity of $X$ is equivalent to the closedness of $\profilex(x)$ for every $x\in\ist$
(see for instance \citealt[\S\,7.1.4]{klein-thompson:1984}).
That~is, for every sequence of profiles $U_n\ (n=1,2,...)$ that converges to a limit $U$,
if~$U_n \in \profilex(x)$ for all $n$, then $U \in\profilex(x).$

\smallskip
Obviously, the preference scores $V_{xy}$ depend continuously on the profile $U$.
Therefore, for the methods that are based on the Llull matrix
it suffices to check for the upper semicontinuity with respect to this matrix.
For such methods we will denote by $\llullx(x)$ the set of Llull matrices $V$ for which $x$ is a choice by the method under consideration.


\medskip
\begin{theorem}
\label{st:prac-usc}
The path-revised approval choice is upper semicontinuous.
\end{theorem}
\begin{proof}\hskip.5em
Let us assume that $\llullx(x)\ni V_n\rightarrow V.$
We want to see that $V\in\llullx(x).$
This is easily checked by considering the inequalities
$(V_n)^*_{x\dft} - (V_n)^*_{\dft x} \ge (V_n)^*_{y\dft} - (V_n)^*_{\dft y}$
that appear in the definition of $x$ being a path-revised approval choice,
and using the fact that the path scores are continuous functions of the original preference scores.
\end{proof}


\paragraph{\textsl{Approval Choice} (\S\,\ref{ssec:should-a-small}) is upper semicontinuous.} \conpar
Let us assume that $\llullx(x)\ni V_n\rightarrow V.$
We want to see that $V\in\llullx(x).$
This is easily checked by considering the inequalities
$(V_n)_{x\dft} - (V_n)_{\dft x} \ge (V_n)_{y\dft} - (V_n)_{\dft y}$
that appear in the definition of $x$ being an approval choice
and letting $n$ tend to~$\infty$.

\paragraph{\textsl{Condorcet-Bucklin} (\S\,\ref{ssec:Condorcet-practical-methods}) is not upper semicontinuous.} \conpar
The following is easily checked to be a counterexample, both for the version where majorities are understood in the strict sense
and for that where they are understood in the weak sense:
$3+\epsilon:~c\better a\better b\aprbar,\ \ 4-\epsilon:~b\better a\better c\aprbar,\ \ 1:~a\better c\better b\aprbar.$

\paragraph{\textsl{Instant Runoff} (\S\,\ref{ssec:apv-instant-runoff}) is upper semicontinuous.} \conpar
Let us assume that $\profilex(x)\ni U_n\rightarrow U.$
We want to see that $U\in\profilex(x).$
For every $n$ there exists a sequence of $N-1$ eliminations that leads to $x$, where $N$ denotes the number of options. Since the sequences of $N-1$ options form a finite set, one can assume ---by extracting a subsequence of the $U_n$--- that the sequence of eliminations is the same for all $n$. Let $a$ be the first eliminated option, $b$ the second one, and so on.

Option~$a$ being eligible for the first elimination means that $\phi(a;U_n) \le \phi(z;U_n)$ for all $z\in A\setminus\{a\}$, where $\phi(z;U)$ stands for the number of first places of $z$ as a function of the profile $U$. Now, for every $z\in A,$ $\phi(z;U)$ depends continuously on $U$. Therefore, by letting $n\rightarrow\infty$ we get $\phi(a;U) \le \phi(z;U)$ for all $z\in A\setminus\{a\}$.

Similarly, option~$b$ being eligible for the second elimination means that $\phi(b\,|\,a;U_n) \le \phi(z\,|\,a;U_n)$ for all $z\in A\setminus\{a,b\}$, where $\phi(z\,|\,a;U)$ stands for the number of first places of $z$ after having eliminated $a$ as a function of the profile $U$. Again the continuous dependence on the profile allows to derive the same inequality for the limiting profile~$U$.

By repeating the same argument, we see that the limiting profile $U$ admits the same sequence of eliminations, and therefore admits $x$ as a choice.

\paragraph{\textsl{Path-Top Choice} (\S\,\ref{ssec:osi-ptc}) is upper semicontinuous.} \conpar
Let $T(V)$ denote the set of path-top choices as a function of the Llull matrix~$V$.
Let us assume that $\llullx(x)\ni V_n\rightarrow V.$
We want to see that $V\in\llullx(x),$ that is, $x\rkrelsub{V}y$ for any $y\neq x$
(where the subindex indicates which Llull matrix are we talking about).
We know that $x\rkrelsub{V_n}y$.
This means that for every $n$ there exists a path $x_0x_1\dots x_m$
of length $m \le N$ ($N$ being the number of options)
 from $x_0=x$ to $x_m=y$
such that $(V_n)^*_{x_ix_{i+1}} \ge (V_n)^*_{x_{i+1}x_i}$ for~all $i$.
The path in question may depend on $n$. However, since the possible paths 
are finite in number, we can assume ---by extracting a subsequence---
that we are dealing with the same path for all $n$.
Now, since the path scores are continuous functions of the original preference scores,
the preceding non-strict inequalities remain true in the limit $n\rightarrow\infty$,
which ensures that $x\rkrelsub{V}y$, as we wanted to show.

\paragraph{\textsl{Strong Path-Top Choice} (\S\,\ref{ssec:osi-ptc}) is upper semicontinuous.} \conpar
Let $S(V)$ denote the set of strong path-top choices as a function of the Llull matrix~$V$.
Let us assume that $\llullx(x)\ni V_n\rightarrow V.$
We want to see that $V\in\llullx(x),$ that is, $x\gerelsub{V}y$ for any $y\neq x$.
We know that $x\gerelsub{V_n}y$.
This means that for every $n$ one has $(V_n)^*_{xy} \ge (V_n)^*_{yx}$.
Now, since the path scores are continuous functions of the original preference scores,
the preceding non-strict inequalities remain true in the limit $n\rightarrow\infty$,
which ensures that $x\gerelsub{V}y$, as we wanted to show.

\paragraph{The \textsl{Swiss Procedures} (\S\,\ref{ssec:Swiss-procedure}) are not upper semicontinuous.} \conpar
As we saw in Section~1.7, \eqref{eq:ex4e} is a counterexample if majorities are understood in the strict sense,
and \eqref{eq:ex4} is a counterexample both for the version where majorities are understood in the strict sense
and for that where they are understood in the weak sense.

\paragraph{\textsl{Preference-Approval Voting} (\S\,\ref{ssec:pba}) is not upper semicontinuous.} \conpar
Examples \eqref{eq:ex4e} and \eqref{eq:ex4} are again counterexamples.

\paragraph{\textsl{Approval Voting with a Runoff} (\S\,\ref{ssec:pba}) is upper semicontinuous.} \conpar
Let us assume that $\llullx(x)\ni V_n\rightarrow V.$
We want to see that $V\in\llullx(x).$
We know that, for every $n$, there exists $z\neq x$ such that $(V_n)_{x\dft}-(V_n)_{\dft x} \ge (V_n)_{t\dft}-(V_n)_{\dft t}$ and $(V_n)_{z\dft}-(V_n)_{\dft z} \ge (V_n)_{t\dft}-(V_n)_{\dft t}$ for all $t\notin\{\,x,z\,\}$; besides, $(V_n)_{xz} \ge (V_n)_{zx}.$ In principle $z$ may depend on~$n$. However, by extracting a subsequence we can assume that it is the same for all~$n$. Once more, the desired result follows by letting $n\rightarrow\infty$.

\subsection{Pareto consistency}
\label{ssec:osi-pareto}

\noindent
As in Section~\ref{ssec:pareto},
\textit{here we make the standing assumption that the votes have the form of truncated rankings with the possibility of ties}.
We also assume that such votes are interpreted as in the following way (see Section~\ref{ssec:ties-and-incomplete}):
a~vote where two options $x$ and $y$ are really ranked equally with each other is
assimilated to half a vote where $x$ is preferred to $y$ together with half another vote with the reverse preference;
any unranked option is less preferred than any ranked one;
nothing is inferred about the preference between two unranked options.

An option $x$ will be said to \dfc{dominate} another one $y$ in the sense of Pareto when
every voter either prefers $x$ to $y$ or ranks them equally, but at least one voter strictly prefers $x$ to $y$.
By definition, a choice method being \dfc{strongly Pareto consistent} means that
it never chooses an option that is dominated by another.
In contrast, \dfc{weak Pareto consistency} means the following: if $x$ dominates $y$
and $y$ is a possible choice, then $x$ is also a possible choice.

\begin{lemma}
\label{st:pareto-lem}
The following inequalities hold whenever $x$ dominates $y$ in the sense of Pareto:
\begin{alignat}{3}
\label{eq:xy}
V_{xy} &> V_{yx} &&&& \\
\label{eq:xa}
V_{xa} &\ge V_{ya},&\quad V_{ax} &\le V_{ay},\qquad&&\text{for any }a\in A\setminus\{\,x,y\,\}. \\
\label{eq:xystar}
V^*_{xy} &\ge V^*_{yx} &&&& \\
\label{eq:xastar}
V^*_{xa} &\ge V^*_{ya},&\quad V^*_{ax} &\le V^*_{ay},\qquad&&\text{for any }a\in A\setminus\{\,x,y\,\}.
\end{alignat}
\end{lemma}
\begin{proof}
Consider all possibilities for the preferences about $x,y$ and $a$ in a truncated ranking that expresses either $x>y$ or $x\sim y$. Altogether, there are nine such possibilities, which are collected in Table~\ref{table:pareto}. Notice that the truncated rankings might not show explicitly all the information that is given here; for instance, a truncated ranking that says ``$x>y$'' without mentioning $a$ belongs to possibility~1; in particular, possibility~6 corresponds to a truncated ranking that mentions neither $y$ nor $a$. For each possibility, the table shows the contribution of a vote of that kind to each of the preference scores $V_{xa},V_{ax},V_{ya}$ and $V_{ay}$. Therefore, if $\alpha_k\ge0$ denotes the number of votes of type $k$, one has $V_{xa} = \sum_k \alpha_k \smash{v^{(k)}_{xa}},$ and analogously for $V_{ax},V_{ya}$ and~$V_{ay}$.

\medskip

\newcommand\strutt{\rule[-5pt]{0pt}{16.5pt}}
\newcommand\struttt{\rule[-5pt]{0pt}{20pt}}

\begin{table}[h]
\begin{center}
\begin{tabular}{|c|c|c|c|c|c|}
\hline\struttt
$k$& preferences &$v^{(k)}_{xa}$&$v^{(k)}_{ya}$&$v^{(k)}_{ax}$&$v^{(k)}_{ay}$\cr
\hline\strutt
1& $x  >  y  >  a$ &1 &1 &0 &0\cr
\hline\strutt
2& $x  >  y\sim a$ &1 &\onehalf &0 &\onehalf \cr
\hline\strutt
3& $x  >  a  >  y$ &1 &0 &0 &1 \cr
\hline\strutt
4& $a\sim x  >  y$ &\onehalf &0 &\onehalf &1 \cr
\hline\strutt
5& $a  >  x  >  y$ &0 &0 &1 &1 \cr
\hline\strutt
6& $x  >  y,a$     &1 &0 &0 &0 \cr
\hline\strutt
7& $x\sim y  >  a$ &1 &1 &0 &0 \cr
\hline\strutt
8& $x\sim a\sim y$ &\onehalf &\onehalf &\onehalf &\onehalf \cr
\hline\strutt
9& $a  >  x\sim y$ &0 &0 &1 &1 \cr
\hline
\end{tabular}
\end{center}
\captionsetup{width=10cm}
\caption{The nine possibilities for the preferences about $x,y$ and $a$
when $x$~dominates $y$ in the sense of Pareto.}
\label{table:pareto}
\end{table}

The inequalities~\eqref{eq:xa} are an immediate consequence of the fact that they hold for every~$k$.
Let us consider now the first of the inequalities~\eqref{eq:xastar}. In order to arrive at it, one can proceed in the following way,
\begin{equation}
\begin{split}
V^*_{ya} &= \min\,(V_{y{a_1}},V_{{a_1}{a_2}},\dots,V_{{a_{n-1}}a})\\
 &\le\, \min\,(V_{x{a_1}},V_{{a_1}{a_2}},\dots,V_{{a_{n-1}}a}) \,\le\, V^*_{xa},
\end{split}
\end{equation}
where \,$y\,a_1 a_2 \dots a_{n-1} a$\, is a path that realizes the maximum that defines $V^*_{ya}$ ---see equation~\eqref{eq:vstar}---
and we have used the first of the inequalities~\eqref{eq:xa}. An analogous argument yields the second of the inequalities~\eqref{eq:xastar}.
Finally, inequality \eqref{eq:xystar} can be obtained in the following way:
if $V^*_{yx} = V_{yx},$ it suffices to notice that $V^*_{yx} = V_{yx} < V_{xy} \le V^*_{xy},$ where the central strict inequality is an immediate consequence of the hypothesis that $x$ dominates $y$ in the sense of Pareto; otherwise, one can write
\begin{equation}
\begin{split}
V^*_{yx} &=\, \min\,(V_{y{a_1}},V_{{a_1}{a_2}},\dots,V_{{a_{n-1}}x})\\
 &\le\, \min\,(V_{x{a_1}},V_{{a_1}{a_2}},\dots,V_{{a_{n-1}}y}) \,\le\, V^*_{xy},
\end{split}
\end{equation}
where \,$y\,a_1 a_2 \dots a_{n-1} x$\, is a path that realizes the maximum that defines $V^*_{yx}$,
and we have used both inequalities~\eqref{eq:xa}.
\end{proof}


\textit{In the remainder of this section we assume that $x$ dominates $y$ in the sense of Pareto}.

\bigskip
\begin{theorem}
\label{st:prac-pareto}
The set of path-revised approval choices is weakly Pareto consistent.
\end{theorem}
\begin{proof}
It suffices to notice that the inequalities \eqref{eq:xastar} with $a=0$ entail the inequality
\begin{equation}
V^*_{x0} - V^*_{0x} \,\ge\, V^*_{y0} - V^*_{0y}.
\end{equation}
This inequality ensures that $y$ cannot be chosen without $x$ being chosen too.
\end{proof}

\paragraph{\textsl{Approval Choice} (\S\,\ref{ssec:should-a-small}) is weakly Pareto consistent.} \conpar
The inequalities \eqref{eq:xa} with $a=0$ result in $V_{x\dft}-V_{\dft x} \ge V_{y\dft}-V_{\dft y},$
which entais that $y$ cannot be chosen without $x$ being chosen too.
The strong version of Pareto consistency is not satisfied since the Approval Choice disregards any preferential information.
 
\paragraph{\textsl{Condorcet-Bucklin} \,(\S\,\ref{ssec:Condorcet-practical-methods}) \,is \,strongly \,Pareto \,consistent.}
Condorcet-Bucklin assumes truncated rankings \textit{without ties}. The hypothesis of the Pareto condition entails that $y$ cannot have first placings at all. The only way for $y$ to win is having a majority of $k$-th placings or better for some $k\ge 2$, but in that case $x$ will have a majority of $(k-1)$-th placings or better.

\paragraph{\textsl{Instant Runoff} (\S\,\ref{ssec:apv-instant-runoff}) is strongly Pareto consistent.} \conpar
Again, Instant Runoff assumes truncated rankings without ties, which entails that $y$ cannot have any first placings. So $y$ is one of the candidates to be eliminated. It the eliminated option is another one, it must be also an option with no first placings. After this elimination we are again in the situation where $y$ has no first placings. Since there are only a finite number of options, $y$ will eventually be the only option with no first placings, and then it will be eliminated.

\paragraph{\textsl{Path-Top Choice} (\S\,\ref{ssec:osi-ptc}) is weakly Pareto consistent.} \conpar
This is an immediate consequence of inequality~\eqref{eq:xystar}.
In fact, according to \S\,\ref{ssec:osi-ptc}, this inequality says that $x\succeq y$.
Therefore, having $y\rkrel z$ for any $z\neq y$ implies $x\rkrel z$ for any $z\neq x$.

\newcommand\vthns{\kern-.1pt\ }
\paragraph{\textsl{Strong\vthns Path-Top\vthns Choice}\vthns (\S\,\ref{ssec:osi-ptc})\vthns is\vthns weakly\vthns Pareto\vthns consistent.} \conpar
See \citealt[\S\,4.3.2]{schulze:2011}.
The weak version of Pareto consistency is an immediate consequence of the inequalities \eqref{eq:xystar} and \eqref{eq:xastar} together with the equality that defines the set of strong path-top choices, namely $S = \{\,x\mid V^*_{xy} \ge V^*_{yx}, \forall y\neq x\,\}$.

As far as we know, there is no proof nor a counterexample of the strong version of Pareto consistency. 

\paragraph{The \textsl{Swiss Procedures} (\S\,\ref{ssec:Swiss-procedure}) are strongly Pareto consistent.} \conpar
Assume that $x$ dominates $y$ in the sense of Pareto. Assume also that $y$ is chosen.
To begin with, this requires $V_{y\dft}-V_{\dft y} \ge 0$. By \eqref{eq:xa}, this implies $V_{x\dft}-V_{\dft x} \ge 0.$
So both $x$ and $y$ and possibly some other non-default option are considered in the second step where the Copeland rule is applied.
Now inequalities \eqref{eq:xy} and \eqref{eq:xa} entail that $x$ wins $y$ under Copeland rule. So $y$ cannot be chosen.

\paragraph{\textsl{Preference-Approval Voting} (\S\,\ref{ssec:pba}) is weakly Pareto consistent.} \conpar
The inequalities \eqref{eq:xa} entail that the approval information will select $x$ whenever it selects $y$.
Now, inequality $V_{xy} > V_{yx}$ ensures that $y$ cannot be a Condorcet winner.
On the other hand, it also entails that $x$ will be in the top cycle whenever $y$ is there.
Therefore we are led back to the approval scores.
Again, the approval information will select $x$ whenever it selects $y$.

The strong version of Pareto consistency is not satisfied. A counterexample is the following:
20~$x \sim y \better z \better t \aprbar$, 20~$t \better x \sim y \better z\aprbar$, 20~$z \better t \better x \sim y \aprbar$, 1~$t \better z \better x \better y \aprbar$. As one can check, the Preference-Approval Voting choice includes all options in spite of the fact that $x$ dominates $y$.

\paragraph{\textsl{Approval Voting with a Runoff} (\S\,\ref{ssec:pba}) is weakly Pareto consistent.} \conpar
Option $y$ being a winner entails that there exists $z$ so that $V_{y\dft} - V_{\dft y} \ge V_{t\dft} - V_{\dft t}$ and $V_{z\dft} - V_{\dft z} \ge V_{t\dft} - V_{\dft t}$ for any $t\neq y,z$. The inequality that is claimed in~A.4.1 shows that either (a)~$z=x$ or (b)~$V_{z\dft} - V_{\dft z} \ge V_{y\dft} - V_{\dft y} = V_{x\dft} - V_{\dft x} \ge V_{t\dft} - V_{\dft t}$ for any $t\neq x,y,z$. In case~(a) the inequality $V_{xy} > V_{yx}$ entails that $x$ is the only winner. In case~(b) the pair $x,z$ is as entitled as $y,z$ for the runoff, and the inequalities \eqref{eq:xa} entail that $V_{xz}-V_{zx} \ge V_{yz}-V_{zy}$, so that $y$ cannot be chosen (from $\{y,z\}$) without $x$ being also chosen (from $\{x,z\}$).

\subsection{Condorcet's last desideratum}
\label{ssec:osi-cld}

In this section we prove the following statement made in \S\,\ref{ssec:Condorcet-last-desideratum},
where the term `confidence' should be understood in the sense of \S\,\ref{ssec:path-scores}:
if $x$ is a path-revised approval choice and $y$ is not,
then the confidence that $x$ is preferable to the default option
is greater than or equal to the confidence that $y$ is preferable to $x$;
besides, this inequality is strict as soon as the path-revised approval choice is a proper one.
This statement is an immediate corollary of the theorem below.

\begin{lemmaN}
\label{st:minInequality}
$V^*_{xz} \,\ge\, \min\,(V^*_{xy}, V^*_{yz})$\, for any pairwise different $x,y,z$.
\end{lemmaN}
\begin{proof}
This follows from the definition of the path scores ---equations (\ref{eq:vstar}--\ref{eq:vgamma})--- by considering the path from $x$ to $z$ which is obtained by concatenating those from $x$ to $y$ and from $y$ to $z$ that produce the respective values of $V^*_{xy}$ and $V^*_{yz}$.
\end{proof}

\xsmallskip
\begin{lemmaN}\ensep
\label{st:laia}
If $V^*_{yz} > V^*_{xz}$, then $V^*_{xz} \ge V^*_{xy}$.\ensep
If $V^*_{xy} > V^*_{xz}$, then $V^*_{xz} \ge V^*_{yz}$.
\end{lemmaN}
\begin{proof}
These implications are immediate consequences of Lemma~\ref{st:minInequality}.
\end{proof}

\xsmallskip
\begin{theorem}
\label{st:prac-topranked}
If $V^*_{x0} - V^*_{0x} \ge 0$ and $V^*_{x0} - V^*_{0x} > V^*_{y0} - V^*_{0y}$,
then $V^*_{x0} - V^*_{0x} \ge V^*_{yx} - V^*_{xy}$.
If the first inequality is strict, then the last one is also strict.
\end{theorem}
\newcommand\ged{\operatorname{\dot\ge}} 
\newcommand\gtd{\operatorname{\dot>}} 
\begin{proof}\hskip.5em
Consider for the moment the first statement, where the first and third inequalities are not strict.
We will argue by contradiction. More specifically, we will arrive at contradiction from the following inequalities
\begin{align}
\label{eq:laia1a}
V^*_{x0} \,&\ged\, V^*_{0x}, \\[2.5pt] 
\label{eq:laia2}
V^*_{x0} + V^*_{0y} \,&>\, V^*_{0x} + V^*_{y0}, \\[2.5pt]
\label{eq:laia3}
V^*_{0x} + V^*_{yx} \,&\gtd\, V^*_{x0} + V^*_{xy}, 
\end{align}
where these inequalities are equivalent respectively to the two inequalities of the hypothesis
and to the negation of the conclusion.
The dots that appear on top of some inequality signs indicate that these signs should be switched from ``$\ge$'' to ``$>$'' and viceversa to prove the statement about strict inequalities.
We will distinguish two cases:\ensep
(a)~$V^*_{x0} \ge V^*_{yx}$;\ensep (b)~$V^*_{yx} > V^*_{x0}$.

Case~(a). Here we are assuming that
\begin{equation}
\label{eq:laia4}
V^*_{x0} \,\ge\, V^*_{yx}.
\end{equation}
This implies that $\min(V^*_{yx}, V^*_{x0}) = V^*_{yx}$. Therefore, by Lemma~\ref{st:minInequality}, 
\begin{equation}
\label{eq:laia5}
V^*_{y0} \,\ge\, V^*_{yx}.
\end{equation}
Now, by concatenating (\ref{eq:laia2}), (\ref{eq:laia5}) and (\ref{eq:laia3}), we get
\begin{equation}
V^*_{x0} + V^*_{0y} \,>\, V^*_{0x} + V^*_{y0} \,\ge\, V^*_{0x} + V^*_{yx} \,\gtd\, V^*_{x0} + V^*_{xy}.
\end{equation}
Therefore,
\begin{equation}
V^*_{0y} \,>\, V^*_{xy},
\end{equation}
and, by Lemma~\ref{st:laia},
\begin{equation}
\label{eq:laia6}
V^*_{xy} \,\ge\, V^*_{x0}.
\end{equation}
Finally, by concatenating (\ref{eq:laia4}), (\ref{eq:laia3}) and (\ref{eq:laia6}), we get
\begin{equation}
V^*_{0x} + V^*_{x0} \,\ge\, V^*_{0x} + V^*_{yx} \,\gtd\, V^*_{x0} + V^*_{xy} \,\ge\, 2V^*_{x0},
\end{equation}
that is $V^*_{0x} \gtd V^*_{x0}$, 
in contradiction with (\ref{eq:laia1a}).

Case~(b). Here we are assuming that
\begin{equation}
\label{eq:laia7}
V^*_{yx} \,>\, V^*_{x0}.
\end{equation}
This implies that $\min(V^*_{yx}, V^*_{x0}) = V^*_{x0}$. Therefore, by Lemma~\ref{st:minInequality}, 
\begin{equation}
\label{eq:laia8}
V^*_{y0} \,\ge\, V^*_{x0}.
\end{equation}
Now, by concatenating (\ref{eq:laia2}) and (\ref{eq:laia8}), we get
\begin{equation}
V^*_{x0} + V^*_{0y} \,>\, V^*_{0x} + V^*_{y0} \,\ge\, V^*_{0x} + V^*_{x0}.
\end{equation}
Therefore,
\begin{equation}
V^*_{0y} \,>\, V^*_{0x},
\end{equation}
and, by Lemma~\ref{st:laia},
\begin{equation}
\label{eq:laia10}
V^*_{0x} \,\ge\, V^*_{yx}.
\end{equation}
However, by combining this inequality with (\ref{eq:laia7}) we get a contradiction with~(\ref{eq:laia1a}).

This completes the proof of the statement with non-strict inequalities.
The reader will easily check that switching between ``$\ge$'' to ``$>$'' in the dotted inequality signs
takes care of the statement with strict inequalities.
\end{proof}

\def\acknowledgement{\par\addvspace{30pt}\small\rmfamily
\trivlist \item[\hskip\labelsep{\bfseries Acknowledgements}]}
\def\endacknowledgement{\endtrivlist\addvspace{-6pt}}
\newenvironment{acknowledgements}{\begin{acknowledgement}}
{\end{acknowledgement}}

\begin{acknowledgements}
This work benefited from our attendance to the \textsl{III Jornadas de Trabajo sobre Sistemas de Votación} (Valdeavellano de Tera, Soria, April 13th, 2013) where a first version was presented. 
We are grateful to the organizers of that workshop for their kind invitation.
\end{acknowledgements}


\end{document}